\documentclass[12pt]{elsarticle}

\usepackage{import} \usepackage{csvsimple}
\usepackage{caption} \usepackage{color}
\usepackage{lscape} \usepackage{afterpage}
\usepackage{pstricks} \usepackage{pst-plot}
\usepackage{longtable} \usepackage{dcolumn}
\usepackage{pst-node} \usepackage{graphicx}
\usepackage{csquotes} \usepackage{amsmath}
\usepackage{multirow} \usepackage{colortab}
\usepackage{color} \usepackage{array}
\usepackage{colortbl} \usepackage{subfigure}
\usepackage{textcomp} \usepackage{pstricks, pst-node}
\usepackage{algorithm,algorithmic} \usepackage{url}
\usepackage{tikz} \usepackage{hyperref} 
\psset{arrows=->, labelsep=3pt, mnode=circle}
\usepackage{eurosym} \usepackage[ampersand]{easylist}
\usepackage{booktabs,caption} \usepackage[flushleft]{threeparttable}
\usepackage{framed}
\usepackage[inline]{enumitem}
\usepackage{varwidth} \usepackage{siunitx}

\usepackage{amssymb}
\usepackage{amsthm}

\usepackage{pgfplots}
\pgfplotsset{width=10cm,compat=1.9}
\usepgfplotslibrary{external}

\newfont{\rams}{msbm10 scaled\magstep1}

\usepackage{relsize,etoolbox}
\AtBeginEnvironment{quote}{\smaller}

\newtheorem{theorem}{Theorem}[section]

\makeatletter
\def\ps@pprintTitle{%
 \let\@oddhead\@empty
 \let\@evenhead\@empty
 \def\@oddfoot{}%
 \let\@evenfoot\@oddfoot}
\makeatother

\begin{document}

\begin{frontmatter}
\title{Discounting and Impatience}

\author[Eco]{\rm Salvatore Greco}
\ead{salgreco@unict.it}
\author[Eco]{\rm Diego Rago}
\ead{diego.rago@phd.unict.it}

\address[Eco]{Department of Economics and Business, University of Catania, Corso Italia, 55, 95129  Catania, Italy}

\date{} \journal{arXiv}

\begin{abstract}

Understanding how people actually trade off time for money 
is perhaps the major question in the field of time discounting.
There is indeed a vast body of work devoted to explore the underlying mechanisms of the individual decision making process in an intertemporal context. 

This paper presents a family of new discount functions whereof we derive a formal axiomatization.
Applying the framework proposed by Bleichrodt, Rohde and Wakker, we further extend their formulation of CADI and CRDI functions, making discounting a function not only of time delay but, simultaneously, also of time distortion.
Our main purpose is, in practice, to provide a tractable setting within which individual intertemporal preferences can be outlined.

Furthermore, we apply our models to study the relation between individual time preferences and personality traits. For the CADI-CADI, results show that the habit of smoking is heavily related with both impatience and time perception. Within the Big-Five framework, conscientiousness, agreeableness and openness are positively related with patience (low r, initial discount rate).

\end{abstract}



\end{frontmatter}




\section{Introduction}

Time discounting has been investigated extensively by researchers, who proposed several models, attempting to replicate how decision makers evaluate intertemporal preferences through the correct functional form\footnote{For an overview, let us mention as valuable literature reviews the works of Frederick, Loewenstein and O'Donoghue \cite{frederick2002time}, Doyle \cite{doyle2013survey} and Cohen et al. \cite{cohen2020measuring}.}.
At a first approximation, the related literature can be condensed into two main paradigms: the exponential discounting and the hyperbolic discounting. The former, developed by Samuelson \cite{samuelson1937note}, assumes individuals discount time at a constant rate: no matter how far from today outcomes are, the discount rate would be same. The exponential model works smoothly for theoretical purposes and it is easy to implement but has a major drawback, as it fails to explain a plethora of inconsistencies that arise from the evidence.

Indeed, it would be more reasonable to assume that people apply higher discount in outcomes closer to the present than outcomes that are further in time, a phenomenon referred to as \emph{decreasing impatience}. That was precisely the insight underlying the hyperbolic discounting. Psychologists studying behavior in both human and animal subjects pioneered hyperbolic discounting, proposing various functional form of the discount functions in order to represent a time decaying discount framework \cite{ainslie1992hyperbolic}. 

Within the hyperbolic framework one can distinguish two sub-families of models: the standard hyperbolic formulation of Loewestein and Prelec \cite{loewenstein1992anomalies} and the quasi-hyperbolic models, firstly proposed by Phelps and Pollak studying intergenerational altruism \cite{phelps1968second}, and later developed by Laibson who analyzed the behavior of hyperbolic consumers in financial decisions \cite{laibson1997golden}. Laibson results suggested that people exhibit preference for immediate reward, revealing what he described as \emph{present-bias}.

Various anomalies \cite{loewenstein1992anomalies} observed in experiments justify and support this class of models. In particular, the \emph{common difference effect} entails a switching preference between two delayed outcomes, when both delays are incremented by a given constant amount: namely, between two rewards close to the present, the majority of population prefers the smaller-sooner but when both rewards are deferred the most are willing to wait choosing the larger-later outcome. Such a behavior seems to uphold the idea that individuals discount the future hyperbolically.

Despite its popularity, the hyperbolic discounting does not allow enough flexibility in its applications. 
In practice, there is no room for increasing impatience or strong decreasing impatience. Indeed, empirical studies observed a \emph{deviation} from the standard decreasing impatience, with individuals exhibiting increasing instead of decreasing impatience, at least at the aggregate level \cite{abdellaoui2010intertemporal}.

In recent years, some researchers developed more sophisticated models that could embed increasing impatience as well. Moreover, most of discounting models had focused on the valuation of monetary outcomes at different points, while little has been said regarding how individuals perceive time. There are models that can accommodate the subjectivity of time \cite{zauberman2009discounting}.

Among the very few studies that make discounting a function of two variables, the work of Read and Scholten is noteworthy. They further develop the \emph{subadditive} Read's model \cite{read2001time}, which claims that declining impatience could arise due to total discounting being greater when the overall time horizon is partitioned into subintervals: in practice, discount rates tend to be higher the closer the outcomes are to one another.
In an elegant paper \cite{scholten2006discounting} Read and Scholten formulate a generalized model of intertemporal choices, the Discounting by Interval model, assuming the discount rate to be a function not only of the time delay, but also of the interval between each future alternative. Their main result was that when the interval between outcomes is very short, discount rate tends to increase with interval length, exhibiting \emph{superadditive} discounting.

Deviations from simple decreasing impatience have been explored by Bleichrodt, Rohde and Wakker \cite{bleichrodt2009non}.
They proposed two classes of discount functions: one with constant absolute (CADI) and the other with constant relative (CRDI) decreasing impatience in order to embed flexibly decreasing or increasing impatience. The CADI class, in comparison with the risky choice models, is presented in fact as the analog of the constant absolute risk averse function (CARA)\footnote{In the same way, the constant relative decreasing impatience is the analog of the constant relative risk aversion in the utility theory.}.

The present work attempts to further develop the discounting framework built by Bleichrodt, Rohde and Wakker.
Our perspective is inspired by the Scholten and Read's insight: namely, making a discount function depending on time delay and time interval, we present a more complete technique to measure time discounting, providing a solid axiomatic architecture.

With discounting being a function of such two variables, we formally derive four families of discount function starting from the Wakker's CADI and CRDI functions.
For example, the CADI-CADI function will be the discount function exhibiting constant absolute decreasing impatience with respect to time delay $t$, and, again, constant absolute decreasing impatience with respect to time interval $T$. In turn, a CRDI-CADI function exhibits constant \emph{relative} decreasing impatience for time delay $t$, but constant \emph{absolute}
decreasing impatience for the time interval $T$. 

The remaining of this paper is organized as follows. 
The formal model is in Section \ref{familyoffunction}. Section \ref{prelec} presents the Prelec's measure for each family of discount function. Next, through a field experiment we empirically test the family of CADI-CRDI models in Section \ref{Experiment}. Finally, Section \ref{conclusion} presents conclusions.




\section{Discounting as a function of both time delay and time interval}\label{familyoffunction}

In this section, we present four classes of discount functions, further extending the  formulation proposed by Wakker et al. concerning CADI and CRDI functions, as dependent not only on time but also on intervals between different times.
We consider the set of time epochs $\mathcal{T}=[0,+\infty[$, a set of outcomes $X$ and the set of timed outcomes $\mathcal{X}=\{(t:x) \in \mathcal{T} \times X\}$. We consider a preference $\succsim$ over $\mathcal{X}$, with, as usually, $\succ$ and $\sim$ representing, respectively, the asymmetric and the symmetric part of $\succsim$. We assume that for all $(t:x),(t+T:y)\in\mathcal{X}$, with $T \ge 0$,
\begin{equation}\label{model1}
(t:x) \succsim (t+T:y) \mbox{ $\iff$ } U(x) \ge F(t,T)\cdot U(y) \notag
\end{equation}
with $U:X \rightarrow \mathbb{R}_{+}$ being a utility function and $F:\mathcal{T}\times \mathbb{R}_+ \rightarrow [0,1]$ being a relative discount function related to the analogous concept introduced in \citep{ok2007theory} for which 
\begin{equation}\label{model2}
(t:x) \succsim (t+T:y) \mbox{ $\iff$ } U(x) \ge \eta(t+T,t)\cdot U(y) \notag
\end{equation}
so that $F(\cdot,\cdot)$ is a reformulation of $\eta(\cdot,\cdot)$ such that for all $t_1,t_2 \in \mathcal{T}, t_1 \le t_2$
\begin{equation}\label{model3}
\eta(t_2,t_1)=F(t_1,t_2-t_1)
\end{equation}
In the following we assume that $U$ is an increasing homeomorphism and $F$ is a continuous map. Preference conditions for the relative discount function $\eta$ that on the basis of (\ref{model3}) can be extended to $F$ have been provided by \citep{ok2007theory}. In the following, as detailed in Section \ref{axioms} we assume that $F(t,T)$ is increasing in $t$ and decreasing in $T$.





\subsection{New classes of discount functions}

We consider the following relative discount functions related to the Constant Absolute Decreasing Impatience (CADI) and Constant Relative Decreasing Impatience (CRDI) introduced in \citep{bleichrodt2009non}, in which $r>0$:
\begin{itemize}
	\item CADI-CADI discount function:
\begin{equation}\label{formula1}
F(t,T)=e^{-\int_0^T r e^{-\delta t} e^{-\gamma \tau} d \tau} 
\end{equation}
that is
\begin{equation}\label{formula2}
F(t,T)=e^{-r e^{-\delta t} \frac{1-e^{-\gamma T}}{\gamma}} \text{ for }\gamma\neq 0
\end{equation}
and 
\begin{equation}\label{formula3}
F(t,T)=e^{-r e^{-\delta t} T}  \text{ for }\gamma=0
\end{equation}
The CADI-CADI discount function is defined for $t\ge 0$ and $T \ge 0$. The monotonicity with respect to $t$ and $T$ implies that $\delta>0$ and $\gamma$ can take any value in $\mathbb{R}$. In the above formulation of relative discount factor $F(t,T)$, the parameter $r$ is a sort of initial discount rate, while, if they assume a positive value, $\delta$ and $\gamma$ are  decay rate of $r$ with respect to $t$ and $T$, respectively. 
As detailed in Section \ref{prelec}, the above mentioned formulations (\ref{formula2}) and (\ref{formula3}) exhibit constant absolute decreasing impatience with respect to time $t$ and  time interval measure $T$.
       \item CRDI-CRDI discount function
\begin{equation}\label{formula7}
F(t,T)=e^{-\int_0^T r t^\alpha \tau^\beta} d \tau
\end{equation}
that is
\begin{equation}\label{formula8}
F(t,T)=e^{-r t^\alpha \frac{T^{\beta+1}}{\beta+1}}
\end{equation}
The CRDI-CRDI discount function is defined for $t>0$ and $T \ge 0$. The monotonicity with respect to $t$ and $T$ implies $\alpha<0$ and $\beta \neq -1$.
In the  formulation (\ref{formula7}) of relative discount factor $F(t,T)$, the parameter $r$ is again an initial discount rate, while $\alpha$ and $\beta$ can be seen as the degree of nonlinear scaling future time perception \cite{kim2013can}  with respect to $t$ and $T$, respectively. 
As detailed in Section \ref{prelec}, the above mentioned formulation (\ref{formula8}) exhibits constant relative decreasing impatience with respect to both time $t$ and time interval measure $T$.

       \item CADI-CRDI discount function
\begin{equation}\label{cadicrdi1}
F(t,T)=e^{-\int_0^T r e^{-\delta t} \tau^\beta} d \tau
\end{equation}
that is
\begin{equation}\label{cadicrdi2}
F(t,T)=e^{-r e^{-\delta t} \frac{T^{\beta+1}}{\beta+1}}
\end{equation}
The CADI-CRDI discount function is defined for $t\ge 0$ and $T \ge 0$.The monotonicity with respect to $t$ and $T$ implies that $\delta>0$ and $\beta \neq -1$.
In the  formulation (\ref{cadicrdi1}) of relative discount factor $F(t,T)$, the parameter $\delta$ is a decay rate of the initial discount rate $r$ with respect to $t$, while  $\beta$ is the degree of nonlinear scaling  time perception  with respect to $T$, respectively.
As detailed in Section The above mentioned formulation (\ref{cadicrdi2}) exhibits constant absolute decreasing impatience with respect to time $t$ and constant relative decreasing impatience for time interval measure $T$.
\item CRDI-CADI discount function
\begin{equation}\label{formula5}
F(t,T)=e^{-\int_0^T r t^\alpha e^{-\gamma \tau}} d \tau
\end{equation}
that is
\begin{equation}\label{formula6}
F(t,T)=e^{-r t^\alpha \frac{1-e^{-\gamma T}}{\gamma}} \text{ for }\gamma\neq 0
\end{equation}
and 
\begin{equation}\label{formula6_bis}
F(t,T)=e^{-r t^\alpha T} \text{ for }\gamma=0
\end{equation}
The CRDI-CADI discount function is defined for $t> 0$ and $T \ge 0$. The monotonicity with respect to $t$ and $T$ implies $\alpha<0$ and $\gamma$ can take any value in $\mathbb{R}$. In a similar manner with respect to the previous case (\ref{cadicrdi1}), this family of discount function $F(t,T)$ presents a decay rate $\gamma$ of the discount rate with respect to $T$ and a nonlinear scaling  time perception $\alpha$ with respect to $t$.
Besides, CRDI-CADI family exhibits constant absolute decreasing impatience with respect to $T$, with the $\gamma$ parameter.
\end{itemize}




\subsection{Structural axioms}\label{axioms}

To characterize the above discount functions we introduce the following axioms:
\begin{itemize}

 \item \textbf{Axiom 1.}  \emph{Time interval monotonicity.} For all $T_1, T_2$ and for all $t$, 
\[
T_1 < T_2 \Longrightarrow F(t, T_1)>F(t, T_2)
\]

As the time span $T$ increases, an outcome further in time is more heavily discounted. This means that, the discount factor gets smaller and such an outcome would be, in fact, less appealing. Formally, the discount factor is, monotonically, strictly decreasing with respect to the variable $T$, the time span. As in our framework discounting is a two-variable function, perhaps there is need to further illustrate this point. 
Let us consider a decision maker that has to deal with binary trade-off decisions between positive monetary outcomes. She can pick the monetary outcome already available in $t$ or, alternatively, she can cash in the monetary outcome in two specific future times: in $t + T_1$ or in $t + T_2$. If $T_1 < T_2$,  the monetary outcome in $t + T_1$, would be available earlier and, therefore, would be less discounted if compared to the outcome available in a later time, as in $t + T_2$.
The reason is indeed that a larger time span $T$ elicits larger discount and, therefore, entails outcomes more distant from a certain time $t$ to be more heavily discounted.
 \item \textbf{Axiom 2.} \emph{Delay monotonicity.} For all $t_1, t_2$ and for all $T$, 
 \[
 t_1 < t_2 \Longrightarrow F(t_1, T)<F(t_2, T)
 \]

 Let us now rather focus on the delay $t$, considering situations whereby there is a different starting point but the same time span $T$. This axiom posits that the discount factor is, monotonically, strictly increasing as the delay increases. Let us suppose that a decision maker is indifferent between receiving the outcome $x$ today or $y$ in a week from now. Delay monotonicity axiom implies that the same decision maker will prefer receiving $x$ tomorrow over receiving $y$ in a week and one day.
	\item \textbf{Axiom 3.}\label{Axiom_3_intro} \emph{Time interval CADI condition}. For all $T_1,T_2$ and $T_3$, for all $t$ and for all $\rho$
	\[
\frac{F(t,T_1)}{F(t,T_2)}=\frac{F(t,T_2)}{F(t,T_3)} \Longrightarrow \frac{F(t,T_1+\rho)}{F(t,T_2+\rho)}=\frac{F(t,T_2+\rho)}{F(t,T_3+\rho)} 
\]

This axiom reflects an additivity condition for discount function ratios. In a CADI framework for the time span $T$, the equality of discount function ratios should be not affected if it is introduced a common extra time span of size $\rho$ for each time interval of the discount functions, remaining unchanged all delays.

\item \textbf{Axiom 3'.}\label{Axiom_3bis_intro} \emph{Time interval CRDI condition}. For all $T_1,T_2$ and $T_3$, for all $t$ and for all $\rho>0$
	\[
\frac{F(t,T_1)}{F(t,T_2)}=\frac{F(t,T_2)}{F(t,T_3)} \Longrightarrow \frac{F(t,T_1\cdot\rho)}{F(t,T_2\cdot\rho)}=\frac{F(t,T_2\cdot\rho)}{F(t,T_3\cdot\rho)} 
\]

Let us now introduce one of the fundamental properties that characterizes the CRDI discounting. The explanation is similar the previous one, but with here scalar multiplication of the time intervals by factor $\rho$. Again, the equality of discount function ratios is not affected after the introduction of the common factor $\rho$.

The following axioms 4 and 4' introduce, respectively, the CADI additivity and CRDI multiplicative conditions applied to the delay $t$. The interpretation is similar to that of 3 and 3' and it can be derived from there.

\item \textbf{Axiom 4.}\label{Axiom_4_intro} \emph{Delay CADI condition}. For all $t_1,t_2$ and $t_3$, for all $T$ and for all $\sigma$
	\[
\frac{F(t_1,T)}{F(t_2,T)}=\frac{F(t_2,T)}{F(t_3,T)} \Longrightarrow \frac{F(t_1+\sigma,T)}{F(t_2+\sigma,T)}=\frac{F(t_2+\sigma,T)}{F(t_3+\sigma,T)} 
\]
\item \textbf{Axiom 4'.}\label{Axiom_4bis_intro} \emph{Delay CRDI condition}. For all $t_1,t_2$ and $t_3$, for all $T$ and for all $\sigma>0$
	\[
\frac{F(t_1,T)}{F(t_2,T)}=\frac{F(t_2,T)}{F(t_3,T)} \Longrightarrow \frac{F(t_1\cdot\sigma,T)}{F(t_2\cdot\sigma,T)}=\frac{F(t_2\cdot\sigma,T)}{F(t_3\cdot\sigma,T)} 
\]
\item \textbf{Axiom 5.} \emph{Weak Total Delay CADI condition}. For all $t_1,t_2,T_1,T_2,\sigma$, 
\[
F(t_1,T_1)=F(t_2,T_2) \Longrightarrow F(t_1+\sigma,T_1)=F(t_2+\sigma,T_2)
 \] 

So far we considered different values of just one variable, the delay or, in turns, the time span. Let us consider to make varies them both. This property posits that when two discount functions with different delay and interval are equal, such an equality holds even when is added a common factor $\sigma$ to both delays.
 
\item \textbf{Axiom 5'.} \emph{Weak Total Delay CRDI condition}. For all $t_1,t_2,T_1,T_2$ and for all $\sigma>0$, 
\[
F(t_1,T_1)=F(t_2,T_2) \Longrightarrow F(t_1\cdot\sigma,T_1)=F(t_2\cdot\sigma,T_2)
 \]

Analogously, in case of an equality between two discount functions, when each delay is multiplied by a common factor $\sigma$ the equality holds.

\item \textbf{Axiom 6.} \emph{Motionless Time interval}. For all $t$,  $F(t,0)=1$.

There is little need to motivate the latter property. In the event of $T=0$, there is no time span at all between two outcomes, namely the latter are referred to the same time and the discount is not triggered amid the two outcomes. Thus, the discount factor equals one.
\item \textbf{Axiom 7.} \emph{Squeeze Delay}. For all $T$, $\lim_{t \to + \infty} F(t,T)=1$.

\end{itemize}




\subsection{Analytical derivation of the axioms}

The above axioms can be reformulated in terms of the preference $\succsim$ on the set of timed outcomes $\mathcal{X}$ as follows.
\begin{itemize}

%
%

\item \textbf{Axiom 1.} For all $T_1, T_2$, for all $t$ and for all $x,y \in X$, if $T_1 < T_2$ 
$$(t : x)\sim(t+T_1 : y) \text{ implies } (t : x) \succ (t+T_2: y).$$
Indeed, $(t : x)\sim(t+T_1 : y)$ means $u(x)=F(t,T_1)\cdot u(y)$ and $(t : x) \succ (t+T_2: y)$ means $u(x) > F(t,T_2)\cdot u(y)$, from which we get $F(t,T_1)>F(t,T_2)$. 

%
%

\item \textbf{Axiom 2.} For all $t_1, t_2$, for all $T$ and for all $x,y \in X$, if $t_1 <t_2$, 
$$(t_1 :x) \sim (t_1+T : y) \text{ implies } (t_2 : x) \succ (t_2+T : y).$$
Indeed, $(t_1 :x) \sim (t_1+T : y)$ means $u(x)=F(t_1,T)\cdot u(y)$ and $(t_2 : x) \succ (t_2+T : y)$ means $u(x) > F(t_2,T)\cdot u(y)$, from which we get $F(t_1,T)<F(t_2,T)$.

%
%

\item \textbf{Axiom 3.} For all $x,y,z,t$ and $\rho$, 
	if
\[(t : z) \sim (t + T : y),
\]
\[(t : w) \sim (t + T : z),
\]
\[(t : x) \sim (t + T_1 : y),
\]
\[ (t : x) \sim (t + T_2 : z),
\]
\[ (t : x) \sim (t + T_3 : w),
\]
\[(t : x)  \sim (t + T_1+\rho : \bar{y}),
\]
\[(t : x)  \sim (t + T_2+\rho : \bar{z}),
\]
\[(t : x)  \sim (t + T_3+\rho : \bar{w}),
\]
\[(t : \bar{z})  \sim (t + \bar{T}  : \bar{y}),
\]
then
\[(t : \bar{w})  \sim (t + \bar{T}  : \bar{z}).
\]
Indeed, 
\begin{itemize}
	\item $(t : z) \sim (t + T : y)$ means $u(z)=F(t,T) \cdot u(y)$, and, consequently, 
\begin{equation}\label{expr1}
F(t,T)=\frac{u(z)}{u(y)} 
\end{equation}
\item $(t : w) \sim (t + T : z)$ means $u(w)=F(t,T) \cdot u(z)$, and, consequently, 
\begin{equation}\label{expr2}
F(t,T)=\frac{u(w)}{u(z)} 
\end{equation}
\item $(t : x) \sim (t + T_1 : y)$ means $u(x)=F(t,T_1) \cdot u(y)$, and, consequently, 
\begin{equation}\label{expr3}
F(t,T_1)=\frac{u(x)}{u(y)} 
\end{equation}
\item $(t : x) \sim (t + T_2 : z)$ means $u(x)=F(t,T_2) \cdot u(z)$, and, consequently, 
\begin{equation}\label{expr4}
F(t,T_2)=\frac{u(x)}{u(z)} 
\end{equation}
\item $(t : x) \sim (t + T_3 : w)$ means $u(x)=F(t,T_3) \cdot u(w)$, and, consequently, 
\begin{equation}\label{expr5}
F(t,T_3)=\frac{u(x)}{u(w)} 
\end{equation}
\item $(t : x)  \sim (t + T_1+\rho : \bar{y})$ means $u(x)=F(t,T_1+\rho) \cdot u(\bar{y})$, and, consequently, 
\begin{equation}\label{expr6}
F(t,T_1+\rho)=\frac{u(x)}{u(\bar{y})} 
\end{equation}
\item $(t : x)  \sim (t + T_2+\rho : \bar{z})$ means $u(x)=F(t,T_2+\rho) \cdot u(\bar{z})$, and, consequently, 
\begin{equation}\label{expr7}
F(t,T_2+\rho)=\frac{u(x)}{u(\bar{z})} 
\end{equation}
\item $(t : x)  \sim (t + T_3+\rho : \bar{w})$ means $u(x)=F(t,T_3+\rho) \cdot u(\bar{w})$, and, consequently, 
\begin{equation}\label{expr8}
F(t,T_3+\rho)=\frac{u(x)}{u(\bar{w})} 
\end{equation}
\item $(t : \bar{z})  \sim (t + \bar{T}  : \bar{y})$ means $u(\bar{z})=F(t,\bar{T}) \cdot u(\bar{y})$, and, consequently, 
\begin{equation}\label{expr9}
F(t,\bar{T})=\frac{u(\bar{z})}{u(\bar{y})} 
\end{equation}
\item $(t : \bar{w})  \sim (t + \bar{T}  : \bar{z})$ means $u(\bar{w})=F(t,\bar{T}) \cdot u(\bar{z})$, and, consequently, 
\begin{equation}\label{expr10}
F(t,\bar{T})=\frac{u(\bar{w})}{u(\bar{z})} 
\end{equation}
\end{itemize}
From (\ref{expr3}) and (\ref{expr4}) we get
\begin{equation}\label{expr11}
\frac{F(t,T_1)}{F(t,T_2)}=\frac{u(z)}{u(y)} 
\end{equation}
From (\ref{expr4}) and (\ref{expr5}) we get
\begin{equation}\label{expr12}
\frac{F(t,T_2)}{F(t,T_3)}=\frac{u(w)}{u(z)} 
\end{equation}
By (\ref{expr1}) and (\ref{expr2}), from (\ref{expr11}) and (\ref{expr12}) we get  
\begin{equation}\label{expr13} \notag
\frac{F(t,T_1)}{F(t,T_2)}=\frac{u(z)}{u(y)}=F(t,T)=\frac{u(w)}{u(z)}=\frac{F(t,T_2)}{F(t,T_3)} 
\end{equation}
From (\ref{expr6}) and (\ref{expr7}) we get
\begin{equation}\label{expr11_1}
\frac{F(t,T_1+\rho)}{F(t,T_2+\rho)}=\frac{u(\bar{z})}{u(\bar{y})} 
\end{equation}
From (\ref{expr7}) and (\ref{expr8}) we get
\begin{equation}\label{expr12_1}
\frac{F(t,T_2+\rho)}{F(t,T_3+\rho)}=\frac{u(\bar{w})}{u(\bar{z})} 
\end{equation}
Taking into account (\ref{expr11_1}) and (\ref{expr12_1})
\begin{equation}\label{expr14}
\frac{F(t,T_1+\rho)}{F(t,T_2+\rho)}=\frac{F(t,T_2+\rho)}{F(t,T_3+\rho)} 
\end{equation}
if and only if 
\begin{equation}\label{expr15}
\frac{u(\bar{z})}{u(\bar{y})}=\frac{u(\bar{w})}{u(\bar{z})} 
\end{equation}
As by (\ref{expr9})  
\begin{equation}\label{expr16}
\frac{u(\bar{z})}{u(\bar{y})}=F(t,\bar{T}) 
\end{equation}
(\ref{expr15}), and, consequently, (\ref{expr14}) require 
\begin{equation}\label{expr17}
\frac{u(\bar{w})}{u(\bar{z})}=F(t,\bar{T}) \notag
\end{equation}
that is, 
\begin{equation}\label{expr18}
(t : \bar{w})  \sim (t + \bar{T}  : \bar{z}). \notag
\end{equation}

%
%

\item \textbf{Axiom 3'.} This axiom represents an analogue case of the previous one, having here the product in place of the sum between $T$, the time elapsing between two outcomes, and $\rho$. Hence, as before, for all $x,y,z,t$ and $\rho$, 
	if
\[(t : z) \sim (t + T : y),
\]
\[(t : w) \sim (t + T : z),
\]
\[(t : x) \sim (t + T_1 : y),
\]
\[ (t : x) \sim (t + T_2 : z),
\]
\[ (t : x) \sim (t + T_3 : w),
\]
\[(t : x)  \sim (t + T_1 \cdot \rho : \bar{y}),
\]
\[(t : x)  \sim (t + T_2 \cdot \rho : \bar{z}),
\]
\[(t : x)  \sim (t + T_3 \cdot \rho : \bar{w}),
\]
\[(t : \bar{z})  \sim (t + \bar{T}  : \bar{y}),
\]
then
\[(t : \bar{w})  \sim (t + \bar{T}  : \bar{z}).
\]

Indeed, standing all other inequalities same as the previous axiom, with regard to the ones with $\rho$:
\begin{itemize}
 \item $(t : x)  \sim (t + T_1 \cdot \rho : \bar{y})$ means $u(x)=F(t,T_1 \cdot \rho) \cdot u(\bar{y})$, and, consequently, 
\begin{equation}\label{axiom3bis_expr1}
F(t,T_1 \cdot \rho)=\frac{u(x)}{u(\bar{y})} 
\end{equation}
\item $(t : x)  \sim (t + T_2 \cdot \rho : \bar{z})$ means $u(x)=F(t,T_2 \cdot \rho) \cdot u(\bar{z})$, and, consequently, 
\begin{equation}\label{axiom3bis_expr2}
F(t,T_2 \cdot \rho)=\frac{u(x)}{u(\bar{z})} 
\end{equation}
\item $(t : x)  \sim (t + T_3 \cdot \rho : \bar{w})$ means $u(x)=F(t,T_3 \cdot \rho) \cdot u(\bar{w})$, and, consequently, 
\begin{equation}\label{axiom3bis_expr3}
F(t,T_3 \cdot \rho)=\frac{u(x)}{u(\bar{w})} 
\end{equation}
\end{itemize}

Proceeding in the same manner as before, from (\ref{axiom3bis_expr1}) and (\ref{axiom3bis_expr2}) we get
\begin{equation}\label{axiom3bis_expr4}
\frac{F(t,T_1 \cdot \rho)}{F(t,T_2 \cdot \rho)}=\frac{u(\bar{z})}{u(\bar{y})} 
\end{equation}
From (\ref{axiom3bis_expr2}) and (\ref{axiom3bis_expr3}) we get
\begin{equation}\label{axiom3bis_expr5}
\frac{F(t,T_2 \cdot \rho)}{F(t,T_3 \cdot \rho)}=\frac{u(\bar{w})}{u(\bar{z})} 
\end{equation}
Taking into account (\ref{axiom3bis_expr4}) and (\ref{axiom3bis_expr5})
\begin{equation}\label{axiom3bis_expr6}
\frac{F(t,T_1 \cdot \rho)}{F(t,T_2 \cdot \rho)}=\frac{F(t,T_2 \cdot \rho)}{F(t,T_3 \cdot \rho)} 
\end{equation}
if and only if 
\begin{equation}\label{axiom3bis_expr7} \notag
\frac{u(\bar{z})}{u(\bar{y})}=\frac{u(\bar{w})}{u(\bar{z})} 
\end{equation}

As by (\ref{expr9}) expression (\ref{expr16}) holds, we have that
(\ref{axiom3bis_expr5}), and, consequently, (\ref{axiom3bis_expr6}) require 
\begin{equation}\label{axiom3bis_expr8}
\frac{u(\bar{w})}{u(\bar{z})}=F(t,\bar{T}) \notag
\end{equation}
that is, 
\begin{equation}\label{axiom3bis_expr9}
(t : \bar{w})  \sim (t + \bar{T}  : \bar{z}). \notag
\end{equation}

%
%

\black

\item \textbf{Axiom 4.} For all $t_1,t_2$ and $t_3$, for all $T$, for all $\sigma$ and for all $x,y,z,\bar{x}, \bar{y},\bar{z},\in X$
	if
\[(t : z) \sim (t + T : y),
\]
\[(t : w) \sim (t + T : z),
\]
\[(t_1 : x) \sim (t_1 + T : y),
\]
\[ (t_2 : x) \sim (t_2 + T : z),
\]
\[ (t_3 : x) \sim (t_3 + T : w),
\]
\[(t_1+\sigma : x)  \sim (t_1 + T+\sigma : \bar{y}),
\]
\[(t_2+\sigma : x)  \sim (t_2 + T+\sigma : \bar{z}),
\]
\[(t_3+\sigma : x)  \sim (t_3 + T+\sigma : \bar{w}),
\]
\[(t : \bar{z})  \sim (t + \bar{T}  : \bar{y}),
\]
then
\[(t : \bar{w})  \sim (t + \bar{T}  : \bar{z}).
\]

Indeed,

\begin{itemize}
    \item $(t_1 : x) \sim (t_1 + T : y)$ means $u(x) = F(t_1,T) \cdot u(y)$ and, consequently,
\begin{equation}\label{axiom4exp3}
F(t_1, T) = \frac{u(x)}{u(y)}
\end{equation}
    \item $(t_2 : x) \sim (t_2 + T : z)$ means
\begin{equation}\label{axiom4exp4}
F(t_2, T) = \frac{u(x)}{u(z)}
\end{equation}
    \item $(t_3 : x) \sim (t_3 + T : w)$ means
\begin{equation}\label{axiom4exp5}
F(t_3, T) = \frac{u(x)}{u(w)}
\end{equation}    
    \item $(t_1 + \sigma : x) \sim (t_1 + T + \sigma: \bar{y})$ means
\begin{equation}\label{axiom4exp6}
F(t_1 + \sigma, T) = \frac{u(x)}{u(\bar{y})}
\end{equation}
    \item $(t_2 + \sigma : x) \sim (t_2 + T + \sigma : \bar{z})$ means
\begin{equation}\label{axiom4exp7}
F(t_2 + \sigma, T) = \frac{u(x)}{u(\bar{z})}
\end{equation}
    \item $(t_3 + \sigma : x) \sim (t_3 + T + \sigma : \bar{w})$ means
\begin{equation}\label{axiom4exp8}
F(t_3 + \sigma, T) = \frac{u(x)}{u(\bar{w})}
\end{equation}
    \item $(t : \bar{z})  \sim (t + \bar{T}  : \bar{y})$ means
\begin{equation}\label{axiom4exp9}
F(t, \bar{T}) = \frac{u(\bar{z})}{u(\bar{y})}
\end{equation}

    \item $(t : \bar{w})  \sim (t + \bar{T}  : \bar{z})$ means
\begin{equation}\label{axiom4exp10}
F(t, \bar{T}) = \frac{u(\bar{w})}{u(\bar{z})}
\end{equation}
\end{itemize}

In the same manner as before, from (\ref{axiom4exp3}) and (\ref{axiom4exp4}) we get 
\begin{equation}\label{axiom4_proof_1}
\frac{F(t_1,T)}{F(t_2,T)}=\frac{u(z)}{u(y)} 
\end{equation}

From (\ref{axiom4exp4}) and (\ref{axiom4exp5}) we get
\begin{equation}\label{axiom4_proof_2}
\frac{F(t_2,T)}{F(t_3,T)}=\frac{u(w)}{u(z)}
\end{equation}

By (\ref{expr1}) and (\ref{expr2}), from (\ref{axiom4_proof_1}) and (\ref{axiom4_proof_2}) we get
\begin{equation}\label{axiom4_proof_3}
\frac{F(t_1,T)}{F(t_2,T)}=\frac{u(z)}{u(y)} = F(t,T) = \frac{u(w)}{u(z)} = \frac{F(t_2,T)}{F(t_3,T)}
\end{equation}

Taking into account (\ref{axiom4exp6}) and (\ref{axiom4exp7})
\begin{equation}\label{axiom4_proof_4}
\frac{F(t_1 + \sigma, T)}{F(t_2 + \sigma, T)}=\frac{F(t_2 + \sigma, T)}{F(t_3 + \sigma, T)} 
\end{equation}

if and only if 

\begin{equation}\label{axiom4_proof_5} \notag
\frac{u(\bar{z})}{u(\bar{y})}=\frac{u(\bar{w})}{u(\bar{z})} 
\end{equation}

As by (\ref{axiom4exp9}), relation (\ref{axiom4_proof_4}) requires

\begin{equation}\label{axiom4_proof_6}
\frac{u(\bar{w})}{u(\bar{z})}= F (t, \bar{T}) \notag
\end{equation}

that is

\begin{equation}\label{axiom4_proof_7}
    (t: \bar{w}) \sim (t + \bar{T} : \bar{z} ) \notag
\end{equation}

%
%

\item \textbf{Axiom 4'.} This axiom represents an analogue of the previous one, having here the product in place of the sum between $t$, the time delay and $\sigma>0$.

%
%
\black

\item \textbf{Axiom 5.} For all $t_1,t_2,T_1,T_2,\sigma$ and for all $x,y,\bar{x},\bar{y} \in X$,
if
\[(t_1 : x) \sim (t_1 + T_1 : y),
\]
\[(t_2 : x) \sim (t_2 + T_2 : y),
\]
\[(t_1+\sigma : \bar{x}) \sim (t_1 + T_1+\sigma : \bar{y}),
\]
then
\[(t_2+\sigma : \bar{x}) \sim (t_2 + T_2+\sigma : \bar{y}).
\]
Indeed,

\begin{itemize}
    \item $(t_1 : x) \sim (t_1 + T_1 : y)$ means $u(x) = F(t_1,T_1) \cdot u(y)$ and, consequently,
\begin{equation}\label{axiom5exp1}
F(t_1, T_1) = \frac{u(x)}{u(y)}
\end{equation}
    \item $(t_2 : x) \sim (t_2 + T_2 : y)$ means $u(x) = F(t_2,T_2) \cdot u(y)$ and, consequently,
\begin{equation}\label{axiom5exp2}
F(t_2, T_2) = \frac{u(x)}{u(y)}
\end{equation}
    \item $(t_1 + \sigma : \bar{x}) \sim (t_1 + T_1 + \sigma : \bar{y})$ means
\begin{equation}\label{axiom5exp3}
F(t_1 + \sigma, T_1) = \frac{u(\bar{x})}{u(\bar{y})}
\end{equation}
    \item $(t_2 + \sigma : x) \sim (t_2 + T_2 + \sigma : \bar{y})$ means
\begin{equation}\label{axiom5exp4}
F(t_2 + \sigma, T_2) = \frac{u(\bar{x})}{u(\bar{y})}
\end{equation}
  \end{itemize}

From (\ref{axiom5exp1}) and (\ref{axiom5exp2}) we get

\begin{equation}\label{axiom5_proof1}
    F(t_1,T_1) = F(t_2,T_2)
\end{equation}

and, therefore, from (\ref{axiom5exp3}) and (\ref{axiom5exp4}) we get also that

\begin{equation}\label{axiom5_proof2}
    F(t_1 + \sigma,T_1) = F(t_2 +\sigma,T_2) \notag
\end{equation}
%
%

\item \textbf{Axiom 5'.} For all $t_1,t_2,T_1,T_2,\sigma$ and for all $x,y,\bar{x},\bar{y} \in X$,
if
\[(t_1 : x) \sim (t_1 + T_1 : y),
\]
\[(t_2 : x) \sim (t_2 + T_2 : y),
\]
\[(t_1 \cdot \sigma : \bar{x}) \sim (t_1 + T_1 \cdot \sigma : \bar{y}),
\]
then
\[(t_2 \cdot \sigma : \bar{x}) \sim (t_2 + T_2 \cdot \sigma : \bar{y}).
\]

Indeed,

\begin{itemize}
    \item $(t_1 \cdot \sigma : \bar{x}) \sim (t_1 \cdot \sigma + T_1  : \bar{y})$ means
\begin{equation}\label{axiom5bisexp3}
F(t_1 \cdot \sigma, T_1) = \frac{u(\bar{x})}{u(\bar{y})}
\end{equation}
    \item $(t_2 \cdot \sigma : x) \sim (t_2 \cdot \sigma + T_2  : \bar{y})$ means
\begin{equation}\label{axiom5bisexp4}
F(t_2 \cdot \sigma, T_2) = \frac{u(\bar{x})}{u(\bar{y})}
\end{equation}
  \end{itemize}

As done for the proof of the previous Axiom, let us consider (\ref{axiom5_proof1}). From (\ref{axiom5bisexp3}) and (\ref{axiom5bisexp4}) we get also that

\begin{equation}\label{axiom5bis_proof}
    F(t_1 \cdot \sigma,T_1) = F(t_2 \cdot \sigma, T_2) \notag
\end{equation}

%
%
\black

\item \textbf{Axiom 6.} For all $t$ and for all $x \in X$,  
\[(t : x) \sim (t : x).
\]
Indeed,
\[(t : x) \sim (t : x)\]
means
\begin{equation}\label{axiom6exp1}
F(t, 0) = 1 \notag
\end{equation}

%
%

\item \textbf{Axiom 7.} For all $T$ and all $x,y \in X$ with $x \succ y$, there is a $\bar{t}$ such that for all $t>\bar{t}$ 
\[(t +T: x) \succ (t: y)
\]
Indeed,
\[(t +T: x) \succ (t: y)
\]
means that 
\begin{equation}\label{axiom7exp1}
1>F(t , T ) > \frac{u(y)}{u(x)} \notag
\end{equation}
Consider a sequence $\{y_n\}$ of outcomes $y_n \in X$ such that $x\succ y_n$ and $y_{n+1} \succ y_n$ for all $n$ and $lim_{n \rightarrow \infty}u(y_n)=u(x)$. The sequence $\{ \frac{u(y_n)}{u(x)}\}$ is increasing and $lim_{n \rightarrow \infty}\frac{u(y_n)}{u(x)}=1$. By the continuity of function $F(t,T)$, for each $y_n$ there is $t_n$ such that $F(t_n,T)=\frac{u(y_n)}{u(x)}$. Moreover,by the monotonicity of $F(t,T)$ with respect to $t$, we get that $t_{n+1}>t_n$ for all $n$. Consequently we have 
\begin{equation}\label{axiom7exp2}
1>F(t_n , T ) > \frac{u(y_n)}{u(x)} \notag
\end{equation}
which, using the squeeze theorem, gives $lim_{n \rightarrow \infty}F(t_n,T)=1$. The latter, in turn, gives $lim_{t \rightarrow \infty}F(t,T)=1$. 
  \end{itemize}




\textbf{Theorem.} The formulation of a discount function is
 \begin{enumerate}[label = {(\alph*)}]
 \item \label{thm:CADI_CADI} CADI-CADI if and only if Axioms 1-7 hold,
 \item \label{thm:CRDI_CRDI} CRDI-CRDI if and only if Axioms 1,2,3',4',5',6,7 hold,
 \item \label{thm:CADI_CRDI} CADI-CRDI if and only if Axioms 1-3, 4',5-7 hold,
 \item \label{thm:CRDI_CADI} CRDI-CADI if and only if Axioms 1,2,3',4,5',6,7 hold. 
 \end{enumerate}

\textbf{Proof.} Here we derive for each family of discount function the proof related to the statements of the above-mentioned theorem.

\textbf{\ref{thm:CADI_CADI}} Let us start by proving that if Axioms 1-7 are satisfied, then the discount function \ref{formula1}, the \emph{CADI-CADI}, holds. Observe that, taking into consideration Axiom 1 and Axiom 3, by Theorem 5.3 in \cite{bleichrodt2009non}, for each $t \ge 0$, there exist $g(t)>0$, $k(t)>0$ and $c(t)$ such that 
\begin{itemize}
	\item for $c(t)>0$, 
	$$F(t,T)=k(t)e^{g(t)e^{-c(t)T}} $$
	\item for $c(t)=0$, 
	$$F(t,T)=k(t)e^{-g(t)T}$$
	\item for $c(t)<0$, 
	$$F(t,T)=k(t)e^{-g(t)e^{-c(t)T}}$$
\end{itemize}
By Axiom 6 we get 
\begin{itemize}
	\item for $c(t)>0$, 
	\begin{equation}\label{propr1}
	F(t,0)=k(t)e^{g(t)}=1 \Rightarrow k(t)=e^{-g(t)} \Rightarrow F(t,T)=e^{-g(t)}e^{g(t)e^{-c(t)T}}=e^{g(t)[e^{-c(t)T}-1]}
	\end{equation}
		\item for $c(t)=0$, 
	\begin{equation}\label{propr2}
	F(t,0)=k(t)=1 \Rightarrow F(t,T)=e^{-g(t)T} 
	\end{equation}
	\item for $c(t)<0$, 
	\begin{equation}\label{propr3}
	F(t,0)=k(t)e^{-g(t)}=1\Rightarrow k(t)=e^{g(t)} \Rightarrow F(t,T)=e^{g(t)}e^{-g(t)e^{-c(t)T}}=e^{g(t)[1-e^{-c(t)T}]} 
		\end{equation}
		\end{itemize}
On the basis of \cite{miyamoto1983axiomatization}, by Axiom 2 and Axiom 4, for each $T>0$, there exist $\hat{g}(t)>0$, $\hat{k}(t)>0$ and $\hat{c}(t)$ such that 
\begin{itemize}
	\item for $\hat{c}(T)>0$, 
	$$F(t,T)=\hat{k}(T)e^{\hat{g}(T)e^{\hat{c}(T)t}} $$
	\item for $\hat{c}(T)=0$, 
	$$F(t,T)=\hat{k}(T)e^{\hat{g}(T)t}$$
	\item for $\hat{c}(T)<0$, 
	$$F(t,T)=\hat{k}(T)e^{-\hat{g}(T)e^{\hat{c}(T)t}}$$
\end{itemize}
Observe that for $\hat{c}(T)>0$, 
$$\lim_{t \to + \infty} F(t,T)=\hat{k}(T)e^{\hat{g}(T)e^{\hat{c}(T)t}} = + \infty $$
which is not acceptable by Axiom 7.  
For $\hat{c}(T)=0$,
$$\lim_{t \to + \infty} F(t,T)=\hat{k}(T)e^{\hat{g}(T)t} = + \infty $$
which is again not acceptable by Axiom 7.  
For $\hat{c}(T)<0$,
$$\lim_{t \to + \infty} F(t,T)=\hat{k}(T)e^{-\hat{g}(T)e^{\hat{c}(T)t}} = \hat{k}(T) $$
which respects Axiom 7 under the condition that $\hat{k}(T)=1$.  
Consequently, by Axioms 2, 4 and 7, for each $T>0$, there exist $\hat{g}(T)>0$ and $\hat{c}(T)<0$ such that
\begin{equation}\label{Discountt}
	F(t,T)=e^{-\hat{g}(T)e^{\hat{c}(T)t}} 
	\end{equation}
Now, let us prove that for all $T_1$, $T_2$ and for all $t_1$, there exists $t_2$ such that:
\begin{equation}\label{Discount4} \notag
F(t_1,T_1) = F(t_2, T_2)
\end{equation}
Without loss of generality, suppose that $ 0 < T_1 < T_2$. Hence, by Axiom 1 and Axiom 6
\begin{equation}\label{Discount5} \notag
1 > F(t_1, T_1) > F(t_1, T_2)
\end{equation}
As by Axiom 2, $F(t,T_2)$ is increasing with respect to $t$, by basic assumption $F(t,T_2)$ is also continuous and by Axiom 7 $\lim_{t \to \infty} F(t,T_2) = 1$, then there exists $\overline{t}$ such that:
\begin{equation}\label{Discount6} \notag
F(\overline{t}, T_2) > F(t_1, T_1)  
\end{equation}
Since $F(\overline{t},T_2) > F(t_1,T_1) > F(t_1, T_2)$, by the continuity of $F(t,T)$, we get that there exists $t_2$ such that
\begin{equation}\label{Discount7}
F(t_2,T_2) = F(t_1,T_1)
\end{equation}


Taking $t_1,t_2,T_1,T_2$ such that $F(t_1,T_1)=F(t_2,T_2)$, by Axiom 5 we have $F(t_1+\sigma,T_1)=F(t_2+\sigma,T_2)$, so that, considering (\ref{Discountt}), we obtain
\begin{equation}\label{Discount1}
F(t_1,T_1)=e^{- \hat{g}(T_1)e^{\hat{c}(T_1)t_1}}=e^{- \hat{g}(T_2)e^{\hat{c}(T_2)t_2}}=F(t_2,T_2)
\end{equation}
\begin{equation}\label{Discount2}
	F(t_1+\sigma,T_1)=e^{- \hat{g}(T_1)e^{\hat{c}(T_1)(t_1+\sigma)}}=e^{- \hat{g}(T_2)e^{
\hat{c}(T_2)(t_2+\sigma)}}=F(t_2+\sigma,T_2)
\end{equation}
From (\ref{Discount1}) and (\ref{Discount2}) we obtain:
$$\frac{e^{-\hat{g}(T_2)e^{\hat{c}(T_2)t_2}}}{e^{-\hat{g}(T_1)e^{\hat{c}(T_1)t_1}}}=\frac{e^{-\hat{g}(T_2)e^{\hat{c}(T_2)(t_2+\sigma)}}}{e^{-\hat{g}(T_1)e^{\hat{c}(T_1)(t_1+\sigma)}}}=
\frac{(e^{-\hat{g}(T_2)e^{\hat{c}(T_2)t_2}})^{e^{\hat{c}(T_2)\sigma}}}{(e^{-\hat{g}(T_1)e^{\hat{c}(T_1)t_1}})^{e^{\hat{c}(T_1)\sigma}}}$$ 
which holds only if 
\begin{equation}\label{Discount3} \notag
\hat{c}(T_1)=\hat{c}(T_2) 
\end{equation}
On the basis of the latter
, (\ref{Discount7}) implies that for all $T_1$ and $T_2$ 
\begin{equation}\label{Discount8} \notag
\hat{c}(T_1)=\hat{c}(T_2)= \bar{c}
\end{equation}
with $\bar{c}<0$.
Consequently,
\begin{equation}\label{Discount9}
F(t,T) = e^{-\hat{g}(T)e^{ \bar{c}t}}
\end{equation}
In case $c(t)>0$, from (\ref{propr1}) and (\ref{Discount9}), we get 
\begin{equation}\label{Discount10}
	F(t,T)=e^{g(t)[e^{-c(t)T}-1]}= e^{-\hat{g}(T)e^{ \bar{c}t}}
	\end{equation}
and, consequently,
\begin{equation}\label{Discount10bis}
g(t)[e^{-c(t)T}-1]= -\hat{g}(T)e^{ \bar{c}t}
	\end{equation}
which is equivalent to 
\begin{equation}\label{Discount11} \notag
\frac{g(t)}{e^{ \bar{c}t}}= \frac{-\hat{g}(T)}{e^{-c(t)T}-1}
	\end{equation}
As $\frac{g(t)}{e^{ \bar{c}t}}$ depends on $t$ and $\frac{-\hat{g}(T)}{e^{-c(t)T}-1}$ depends on $t$ and $T$, then
\begin{itemize}
	\item there exists some $k>0$ such that 
\begin{equation}\label{Discount12}
\frac{g(t)}{e^{ \bar{c}t}}= \frac{-\hat{g}(T)}{e^{-c(t)T}-1}=k
	\end{equation}
\item there exists some $\widetilde{c}>0$ such that 
\begin{equation}\label{Discount13}
c(t)=\widetilde{c}
	\end{equation}
\end{itemize}
From (\ref{propr1}), (\ref{Discount12}) and (\ref{Discount13}) we get
\begin{equation}\label{Discount14} \notag
	F(t,T)=e^{ke^{ \bar{c}t[e^{-\widetilde{c}T}-1]}}
	\end{equation}
Taking into account (\ref{formula2}) we get
\begin{equation}\label{Discount15} \notag
F(t,T)=e^{ke^{ \bar{c}t[e^{-\widetilde{c}T}-1]}}=e^{-r e^{-\delta t} \frac{1-e^{-\gamma T}}{\gamma}} 
\end{equation}
with $k=\frac{r}{\gamma}$, $\bar{c}=-\delta$ and $\widetilde{c}=\gamma$ with $r>0, \delta>0$ and $\gamma>0$. 
 
In case $c(t)=0$, from (\ref{propr2}) and (\ref{Discount9}), we get 
\begin{equation}\label{Discount16} \notag
	F(t,T)=e^{-g(t)T}= e^{-\hat{g}(T)e^{\bar{c}t}}
	\end{equation}
and, consequently,
\begin{equation}\label{Discount17} \notag
g(t)T= \hat{g}(T)e^{\bar{c}t}
	\end{equation}
which is equivalent to 
\begin{equation}\label{Discount18} \notag
\frac{g(t)}{e^{\bar{c}t}}= \frac{\hat{g}(T)}{T}
	\end{equation}
As $\frac{g(t)}{e^{\bar{c}t}}$ depends on $t$ and $\frac{\hat{g}(T)}{T}$ depends on $T$, then there exists some $k > 0$ such that 
\begin{equation}\label{Discount20}
\frac{g(t)}{e^{\bar{c}t}}= \frac{\hat{g}(T)}{T}=k
	\end{equation}

From (\ref{propr2}) and (\ref{Discount20}) we get
\begin{equation}\label{Discount21} \notag
	F(t,T)=e^{-ke^{\bar{c}t}T}
	\end{equation}
Taking into account (\ref{formula3}) we get
\begin{equation}\label{Discount22} \notag
F(t,T)=e^{-ke^{\bar{c}t}T}=e^{-r e^{-\delta t} T} 
\end{equation}
with $k=r$ and $\bar{c}=-\delta$ with $r$ and $\delta$ both positive.

In case $c(t)<0$, from (\ref{propr3}) and (\ref{Discount9}), we get 
\begin{equation}\label{Discount23} \notag
	F(t,T)=e^{g(t)[1-e^{-c(t)T}]}= e^{-\hat{g}(T)e^{\bar{c}t}}
	\end{equation}
and, consequently,
\begin{equation}\label{Discount24} \notag
g(t)[1-e^{-c(t)T}]= -\hat{g}(T)e^{\bar{c}t}
	\end{equation}
which is equivalent to 
\begin{equation}\label{Discount25} \notag
\frac{g(t)}{e^{\bar{c}t}}= \frac{-\hat{g}(T)}{1-e^{-c(t)T}}
	\end{equation}
As $\frac{g(t)}{e^{\bar{c}t}}$ depends on $t$ and $\frac{-\hat{g}(T)}{1-e^{-c(t)T}}$ depends on $t$ and $T$, then
\begin{itemize}
	\item there exists some $k>0$ such that 
\begin{equation}\label{Discount26}
\frac{g(t)}{e^{\bar{c}t}}= \frac{-\hat{g}(T)}{1-e^{-c(t)T}}=k
	\end{equation}
\item there exists some $\widetilde{c}<0$ such that 
\begin{equation}\label{Discount27}
c(t)=\widetilde{c}
	\end{equation}
\end{itemize}
From (\ref{propr3}), (\ref{Discount26}) and (\ref{Discount27}) we get
\begin{equation}\label{Discount28} \notag
	F(t,T)=e^{ke^{\bar{c}t[1-e^{-\widetilde{c}T}]}}
	\end{equation}
Taking into account (\ref{formula1}) we get 
\begin{equation}\label{Discount29} \notag
F(t,T)=e^{ke^{\bar{c}t[1-e^{-\widetilde{c}T}]}}=e^{-r e^{-\delta t} \frac{1-e^{-\gamma T}}{\gamma}} 
\end{equation}
with $k=-\frac{r}{\gamma}$, $\bar{c}=\delta$ and $\widetilde{c}=\gamma$, with $r$ and $\delta$ positive and $\gamma$ negative.


\bigskip

\textbf{\ref{thm:CRDI_CRDI}} Let us prove now that if Axioms 1,2,3',4',5',6,7 hold, then the discount function CRDI-CRDI in (\ref{formula7}) holds. Observe that, taking into consideration Axiom 1 and Axiom 3', by Theorem 6.3 in \cite{bleichrodt2009non}, for each $t \ge 0$, there exist $g(t)>0$, $k(t)>0$ and 
\begin{itemize}
	\item for $c(t)>0$, 
	$$F(t,T)=k(t)e^{-g(t)T^{c(t)}} $$
	\item for $c(t)=0$, 
	$$F(t,T)=k(t)T^{-g(t)}$$
	\item for $c(t)<0$, 
	$$F(t,T)=k(t)e^{g(t)T^{c(t)}}$$
\end{itemize}
By Axiom 6 we get that for $c(t)>0$ we get 
	\begin{equation}\label{propr1_}
	F(t,0)=k(t)=1 
	\end{equation}
Instead, for $c(t)=0$ as well as for $c(t)<0$, $F(t,0)$ is not defined. Consequently, $c(t)=0$ and $c(t)<0$ cannot be accepted and we have
	\begin{equation}\label{T_CRDI}
	F(t,0)=F(t,T)=e^{-g(t)T^{c(t)}}  
	\end{equation}

. 
	
	By Axiom 2 and Axiom 4', again by Theorem 6.3 in \cite{bleichrodt2009non}, for each $T>0$, there exist $\hat{g}(T)>0$, $\hat{k}(T)>0$ and $\hat{c}(T)$ such that 
\begin{itemize}
	\item for $\hat{c}(T)>0$, 
	$$F(t,T)=\hat{k}(T)e^{\hat{g}(T)t^{\hat{c}(T)}} $$
	\item for $\hat{c}(T)=0$, 
	$$F(t,T)=\hat{k}(T)t^{\hat{g}(T)}$$
	\item for $\hat{c}(T)<0$, 
	$$F(t,T)=\hat{k}(T)e^{-\hat{g}(T)t^{\hat{c}(T)}}$$
\end{itemize}
Observe that for $\hat{c}(T)>0$, 
$$\lim_{t \to + \infty} F(t,T)=\lim_{t \to + \infty}\hat{k}(T)e^{\hat{g}(T)t^{\hat{c}(T)}}  = + \infty $$
which is not acceptable by Axiom 7.  
For $\hat{c}(T)=0$,
$$\lim_{t \to + \infty} F(t,T)=\lim_{t \to + \infty}\hat{k}(T)t^{\hat{g}(T)}  = + \infty $$
which is again not acceptable by Axiom 7.  
For $\hat{c}(T)<0$,
$$\lim_{t \to + \infty} F(t,T)=\lim_{t \to + \infty} \hat{k}(T)e^{-\hat{g}(T)t^{\hat{c}(T)}} = \hat{k}(T) $$
which respects Axiom 7 under the condition that $\hat{k}(T)=1$.  
Consequently, by Axioms 2, 4' and 7, for each $T>0$, there exist $\hat{g}(T)>0$ and $\hat{c}(T)<0$ such that
\begin{equation}\label{Discountt_}
	F(t,T)=e^{-\hat{g}(T)t^{\hat{c}(T)}} 
	\end{equation}
With the same proof for the case CADI-CADI, we have that for all $T_1$, $T_2$ and for all $t_1$, there exists $t_2$ such that $F(t_1,T_1) = F(t_2, T_2)$. Taking $t_1,t_2,T_1,T_2$ such that $F(t_1,T_1)=F(t_2,T_2)$, by Axiom 5' we have $F(t_1\cdot\sigma,T_1)=F(t_2\cdot\sigma,T_2)$, so that, considering (\ref{Discountt_}), we obtain
\begin{equation}\label{Discount1_}
F(t_1,T_1)=e^{- \hat{g}(T_1)t_1^{\hat{c}(T_1)}}=e^{- \hat{g}(T_2)t_2^{\hat{c}(T_2)}}=F(t_2,T_2)
\end{equation}
\begin{equation}\label{Discount2_}
	F(t_1\cdot\sigma,T_1)=e^{- \hat{g}(T_1)(t_1\cdot\sigma)^{\hat{c}(T_1)}}=e^{- \hat{g}(T_2)(t_2\cdot\sigma)^{\hat{c}(T_2)}}=F(t_2\cdot\sigma,T_2)
\end{equation}
From (\ref{Discount1_}) and (\ref{Discount2_}) we obtain:
$$\frac{e^{\hat{g}(T_2)t_2^{\hat{c}(T_2)}}}{e^{\hat{g}(T_1)t_1^{\hat{c}(T_1)}}}=\frac{e^{\hat{g}(T_2)(t_2\cdot\sigma)^{\hat{c}(T_2)}}}{e^{\hat{g}(T_1)(t_1\cdot\sigma)^{\hat{c}(T_1)}}}=
\frac{(e^{\hat{g}(T_2)t_2^{\hat{c}(T_2)}})^{\sigma^{\hat{c}(T_2)}}}{(e^{\hat{g}(T_1)t_1^{\hat{c}(T_1)}})^{\sigma^{\hat{c}(T_1)}}}$$ 
which holds only if 
\begin{equation}\label{Discount3_}
\hat{c}(T_1)=\hat{c}(T_2) 
\end{equation}
since, as already mentioned, for all $T_1$, $T_2$ and for all $t_1$, there exists $t_2$ such that $F(t_1,T_1) = F(t_2, T_2)$, 
(\ref{Discount3_}) implies that for all $T_1$ and $T_2$ 
\begin{equation}\label{Discount8_} \notag
\hat{c}(T_1)=\hat{c}(T_2)= \bar{c} < 0,
\end{equation}
and, consequently,
\begin{equation}\label{Discount9_} 
F(t,T) = e^{-\hat{g}(T)t^{\bar{c}}}
\end{equation}
From the latter and (\ref{propr1_}), we get 
	\begin{equation}\label{Discount10_1}
	F(t,T)=e^{-g(t)T^{c(t)}}= e^{-\hat{g}(T)t^{\bar{c}}}
	\end{equation}
	with $c(t)>0$ and $\bar{c}<0$. 
	Taking into account (\ref{Discount10_1}), we get
\begin{equation}\label{Discount10_1bis}
g(t)T^{c(t)}=\hat{g}(T)t^{\bar{c}}
	\end{equation}
which is equivalent to 
\begin{equation}\label{Discount11_1bis} \notag
\frac{g(t)}{t^{\bar{c}}}= \frac{\hat{g}(T)}{T^{c(t)}}
	\end{equation}
As $\frac{g(t)}{t^{\bar{c}}}$ depends on $t$ and $\frac{\hat{g}(T)}{T^{c(t)}}$ depends on $t$ and $T$, then
\begin{itemize}
	\item there exists some $k$ such that 
\begin{equation}\label{Discount12_1}
\frac{g(t)}{t^{\bar{c}}}= \frac{\hat{g}(T)}{T^{c(t)}}=k>0
	\end{equation}
\item there exists some $\widetilde{c}>0$ such that 
\begin{equation}\label{Discount13_1}
c(t)=\widetilde{c}
	\end{equation}
\end{itemize}
From (\ref{propr1_}), (\ref{Discount12_1}) and (\ref{Discount13_1}) we get
\begin{equation}\label{Discount14_} \notag
	F(t,T)=e^{-kt^{\bar{c}}T^{\widetilde{c}}}
	\end{equation}
Taking into account (\ref{formula8}) we get
\begin{equation}\label{Discount15_}
F(t,T)=e^{-kt^{\bar{c}}T^{\widetilde{c}}}=e^{-r t^{-\alpha} \frac{T^{\beta+1}}{\beta+1}} 
\end{equation}
with $k=\frac{r}{\beta+1}$, $\bar{c}=-\alpha$ and $\widetilde{c}=\beta+1$, with $r>0$, $\alpha>0$ and $\beta>-1$.

\bigskip

\textbf{\ref{thm:CADI_CRDI}} We prove now that if  Axioms 1-3, 4',5-7 hold then the discount function has a \emph{CADI-CRDI} formulation (\ref{cadicrdi2}).
By Axioms 1, 3' and 7, which imply (\ref{Discountt_}), and by Axioms 2, 4 and 7, which imply (\ref{T_CRDI})
\begin{equation}\label{Discount_CADI_CRDI_1} \notag
	F(t,T)=e^{-\hat{g}(T)e^{\bar{c}t}}=e^{-g(t)T^{c(t)}} 
	\end{equation}
with $\bar{c}<0$ and $c(t)>0$. 	
	
	Consequently,
%

\begin{equation}\label{hybrid_1}
    \hat{g}(T)e^{\bar{c}t} = g(t)T^{c(t)}
\end{equation}
from which we get
\begin{equation}\label{Discount_CADI-CRDI_3} \notag
    \frac{\hat{g}(T)}{T^{c(t)}} = \frac{g(t)}{e^{\bar{c}t}}.
\end{equation}
As $\frac{\hat{g}(T)}{T^{c(t)}}$ depends on $t$ and $T$ and  $\frac{g(t)}{e^{\bar{c}t}}$ depends on $t$, then
\begin{itemize}
	\item there exists some $k$ such that 
\begin{equation}\label{Discount_cadi_crdi_1}
\frac{\hat{g}(T)}{T^{c(t)}} = \frac{g(t)}{e^{\bar{c}t}}=k>0
	\end{equation}
\item there exists some $\widetilde{c}>0$ such that 
\begin{equation}\label{Discount_cadi_crdi_2}
c(t)=\widetilde{c}
	\end{equation}
\end{itemize}
From (\ref{propr1_}), (\ref{Discount_cadi_crdi_1}) and (\ref{Discount_cadi_crdi_2}) we get
\begin{equation}\label{Discount_cadi_crdi_final} \notag
	F(t,T)=e^{-ke^{\bar{c}t}T^{\widetilde{c}}}
	\end{equation}
Taking into account (\ref{formula8}) we get
\begin{equation}\label{Discount15_}
F(t,T)=e^{-ke^{\bar{c}t} T^{\widetilde{c}}}=e^{-r e^{-\delta t} \frac{T^{\beta+1}}{\beta+1}} 
\end{equation}
with $k=\frac{r}{\beta+1}$, $\bar{c}=-\delta$ and $\widetilde{c}=\beta+1$, with $r>0$, $\delta>0$ and $\beta>-1$.

%
%
%
%
%
%
%
%
%
%
%
%


\bigskip

\textbf{\ref{thm:CRDI_CADI}} We prove now that if  Axioms 1-3, 4',5',6,7, then the discount function has a \emph{CRDI-CADI} formulation (\ref{formula6}) or (\ref{formula6_bis}).
By Axioms 1,3 and 6, on the basis of (\ref{propr1}), (\ref{propr2}) and (\ref{propr3}), and  by Axioms 2, 4' and 5' which imply (\ref{Discount9_}), we have the following three cases, with $\bar{c}<0, g(t)>0$ and $\hat{g}(T)>0$,  
\
\begin{itemize}
	\item for $c(t)>0$
	\begin{equation}\label{propr1_}
	F(t,T)=e^{g(t)[e^{-c(t)T}-1]}=e^{-\hat{g}(T)t^{\bar{c}}}
	\end{equation}
		\item for $c(t)=0$, 
	\begin{equation}\label{propr2_}
	F(t,T)=e^{-g(t)T}=e^{-\hat{g}(T)t^{\bar{c}}} 
	\end{equation}
	\item for $c(t)<0$, 
	\begin{equation}\label{propr3_}
	F(t,T)=e^{g(t)[1-e^{-c(t)T}]}=e^{-\hat{g}(T)t^{\bar{c}}} 
		\end{equation}
 \end{itemize}

From (\ref{propr1_}) we get 
\begin{equation}\label{propr1_A}
	g(t)[e^{-c(t)T}-1]=-\hat{g}(T)t^{\bar{c}}
	\end{equation}
 which is equivalent to 
\begin{equation}\label{propr1_B}
\frac{g(t)}{t^{\bar{c}}}=\frac{-\hat{g}(T)}{e^{-c(t)T}-1}
	\end{equation}

As $\frac{g(t)}{t^{\bar{c}}}$ depends on $t$ and $\frac{-\hat{g}(T)}{e^{-c(t)T}-1}$ depends on $t$ and $T$, then
\begin{itemize}
	\item there exists some $k$ such that 
\begin{equation}\label{Discount_1}
\frac{g(t)}{t^{\bar{c}}}=\frac{-\hat{g}(T)}{e^{-c(t)T}-1}=k>0
	\end{equation}
\item there exists some $\widetilde{c}>0$ such that 
\begin{equation}\label{Discount_2}
c(t)=\widetilde{c}
	\end{equation}
\end{itemize}
From (\ref{propr1_}), (\ref{Discount_1}) and (\ref{Discount_2}) we get
\begin{equation}\label{Discount_4} \notag
	F(t,T)=e^{kt^{\bar{c}}[e^{-\widetilde{c}T}-1]}
	\end{equation}
Taking into account (\ref{formula6}) we get
\begin{equation}\label{Discount_5}
F(t,T)=e^{-kt^{\bar{c}}[1-e^{-\widetilde{c}T}]}=e^{-r t^{\alpha}\frac{1-e^{-\gamma T}}{\gamma}} 
\end{equation}
with $k=\frac{r}{\gamma}$, $\bar{c}=\alpha$ and $\widetilde{c}=\gamma$, with $r>0$, $\alpha<0$ and $\gamma>0$.

From (\ref{propr2_}) we get 
\begin{equation}\label{propr2_A}
g(t)T=\hat{g}(T)t^{\bar{c}}
	\end{equation}
 which is equivalent to 
\begin{equation}\label{propr2_B}
\frac{g(t)}{t^{\bar{c}}}=\frac{\hat{g}(T)}{T}
	\end{equation}

As $\frac{g(t)}{t^{\bar{c}}}$ depends on $t$ and $\frac{\hat{g}(T)}{T}
$ depends on $T$, then there exists some $k$ such that 
\begin{equation}\label{Discount_1_}
\frac{g(t)}{t^{\bar{c}}}=\frac{\hat{g}(T)}{T}=k>0.
	\end{equation}
From (\ref{propr2_}) and (\ref{Discount_1_}) we get
\begin{equation}\label{Discount_4_} \notag
	F(t,T)=e^{-kt^{\bar{c}}T}
	\end{equation}
Taking into account (\ref{formula6_bis}) we get
\begin{equation}\label{Discount_5_}
F(t,T)=e^{-kt^{\bar{c}}T}=e^{-r t^\alpha T}  
\end{equation}
with $k=r$, $\bar{c}=\alpha$,  with $r>0$ and  $\alpha<0$.

From (\ref{propr3_}) we get 
\begin{equation}\label{propr3_A}
	g(t)[1-e^{-c(t)T}]=-\hat{g}(T)t^{\bar{c}}
	\end{equation}
 which is equivalent to 
\begin{equation}\label{propr3_B}
\frac{g(t)}{t^{\bar{c}}}=\frac{-\hat{g}(T)}{1-e^{-c(t)T}}
	\end{equation}

As $\frac{g(t)}{t^{\bar{c}}}$ depends on $t$ and $\frac{-\hat{g}(T)}{1-e^{-c(t)T}}$ depends on $t$ and $T$, then
\begin{itemize}
	\item there exists some $k$ such that 
\begin{equation}\label{Discount_1__}
\frac{g(t)}{t^{\bar{c}}}=\frac{-\hat{g}(T)}{1-e^{-c(t)T}}=k>0
	\end{equation}
\item there exists some $\widetilde{c}<0$ such that 
\begin{equation}\label{Discount_2__}
c(t)=\widetilde{c}
	\end{equation}
\end{itemize}
From (\ref{propr3_}), (\ref{Discount_1__}) and (\ref{Discount_2__}) we get
\begin{equation}\label{Discount_4-} \notag
	F(t,T)=e^{kt^{\bar{c}}[1-e^{-\widetilde{c}(t)T}]}
	\end{equation}
Taking into account (\ref{formula6}) we get
\begin{equation}\label{Discount_5}
F(t,T)=e^{kt^{\bar{c}}[1-e^{-\widetilde{c}T}]}=e^{-r t^{\alpha}\frac{1-e^{-\gamma T}}{\gamma}} 
\end{equation}
with $k=-\frac{r}{\gamma}$, $\bar{c}=\alpha$ and $\widetilde{c}=\gamma$, with $r>0$, $\alpha<0$ and $\gamma<0$.

\section{Properties of the discount functions}\label{prelec}

In this section we will derive for each family of discount functions the Prelec's measure \cite{prelec2004decreasing}. It is a criterion playing an analogous role to the Pratt-Arrow measure \citep{pratt1978risk} for the utility function: indeed, in the latter case it identifies the risk aversion, while in intertemporal choice it identifies the dynamic inconsistencies of agents.
For the sake of simplicity, in the following we consider the Pratt–Arrow degree of convexity of the logarithm of the discount function, namely the Prelec's measure known also as \emph{decreasing impatience}, that is defined as:

\begin{equation}\label{Prelec_1}\notag
\lambda (\tau) = - \frac{\big[\ln \ \phi(\tau)\big]''}{\big[\ln \ \phi(\tau) \big]'}
\end{equation}

where $\phi: \mathbb{R}^+ \rightarrow \mathbb{R}^+$ is  a ``classical'' discount function such that $\phi(t)$ gives the present value of one unit (of money or utility, according to the interpretations) available after $t$ years.

Considering a discount function $F(t,T)$, we can define Prelec's measure to determine 
the degree of impatience both with respect to the time $t$ and with respect to the interval $T$, as
\begin{equation}\label{Prelec_t}
\lambda_1 (t,T) = - \frac{\big[\frac{\partial^2\ln \ F(t,T)}{\partial t^2}\big]}{\big[\frac{\partial\ln \ F(t,T)}{\partial t}\big]} \notag
\end{equation}

\begin{equation}\label{Prelec_T}
\lambda_2 (t,T) = - \frac{\big[\frac{\partial^2\ln \ F(t,T)}{\partial T^2}\big]}{\big[\frac{\partial\ln \ F(t,T)}{\partial T}\big]} \notag
\end{equation}


%
\subsection{Constant Absolute Decreasing Impatience for t and T}
%


\begin{theorem} If Axioms 1,2,6 and 7 hold, then a discount function $F(t,T)$ belongs to the CADI-CADI family if and only if  the Prelec's measures with respect to $t$ and $T$ are constant, that is $\lambda_1 (t,T)=k_1$ and $\lambda_2(t,T)=k_2$, with $k_1>0$ and $k_2 \in \mathbb{R}$.
\bigskip
\end{theorem}
\begin{proof} Let us start by proving that the CADI-CADI discount function

$$\\ F(t,T) = e^{- r e^{-\delta t }\frac{1-e^{- \gamma T}}{\gamma} }$$

has constant Prelec's measures $\lambda_1 (t,T)$ and $\lambda_2 (t,T)$. Indeed, since:

\begin{equation}\label{Prelec_2}
\frac{\partial \ln F(t,T)}{\partial t} =  \delta r\frac{1-e^{- \gamma T}}{\gamma} e^{-\delta t  }\notag 
\end{equation}

and 

\begin{equation}\label{Prelec_3}\notag
\frac{\partial^2 \ln F(t,T)}{\partial t^2} =  -\delta^2 r\frac{1-e^{- \gamma T}}{\gamma} e^{-\delta t  }  
\end{equation}

we get 

\begin{equation}\label{Prelec_4}\notag
\lambda_1 (t,T)=\delta
\end{equation}

Analogously, from

\begin{equation}\label{Prelec5}\notag
\frac{\partial \ln F(t,T)}{\partial T} = -r e^{-\delta t} e^{- \gamma T} 
\end{equation}

and

\begin{equation}\label{Prelec6}\notag
\frac{\partial^2 \ln F(t,T)}{\partial T^2} = \gamma r e^{-\delta t} e^{- \gamma T} 
\end{equation}

we get

\begin{equation}\label{Prelec7}\notag
\lambda_2(t,T) = \gamma
\end{equation}

Now we prove that if Axioms 1,2,6 and 7 hold and the Prelec's measures $\lambda_1 (t,T)$ and $\lambda_2 (t,T)$ are constant, then a discount function $F(t,T)$ belongs to the CADI-CADI family. Let us consider a function $\phi: \mathbb{R}^+\rightarrow \mathbb{R}^+$ and a costant $k>0$ such that
\begin{equation}\label{Prelec_8}
\lambda (x) = - \frac{\big[ln \ \phi(x)\big]''}{\big[\ln \ \phi(x) \big]'}=k
\end{equation}

Let us consider the function $z : \mathbb{R}^+\rightarrow \mathbb{R}^+$ defined as
\begin{equation}\label{Prelec_9}
z(x) = \big[ln \ \phi(x)\big]'
\end{equation}

so that 
\begin{equation}\label{Prelec_10}\notag
\lambda (x) = - \frac{z(x)'}{z(x)}
\end{equation}

and we can rewrite (\ref{Prelec_8}) as
\begin{equation}\label{Prelec_11}\notag
\frac{z(x)'}{z(x)}=-k
\end{equation}

and, consequently
\begin{equation}\label{Prelec_12}
\left[ln \ |z(x)|\right]'=-k \notag
\end{equation}

the solution of which is 
\begin{equation}\label{Prelec_13}
z(x)=a e^{-kx}, \mbox{ or } z(x)=-a e^{-kx}, a>0.
\end{equation}

Taking into account (\ref{Prelec_9}) and (\ref{Prelec_13}) we get
\begin{equation}\label{Prelec_14}
\big[ln \ |\phi(x)|\big]'=a e^{-kx} \mbox{ or } \big[ln \ |\phi(x)|\big]'=-a e^{-kx}  
\end{equation}

Since $\phi(x)>0$, the solution of (\ref{Prelec_14}) is
\begin{equation}\label{Prelec_14bis}
\phi(x)=be^{-\frac{a}{k} e^{-kx}} \mbox{ or } \phi(x)=be^{\frac{a}{k} e^{-kx}}, a>0, b>0
\end{equation}
if $k\neq 0$, and
\begin{equation}\label{Prelec_14ter}
\phi(x)=be^{ax} \mbox{ or } \phi(x)=be^{-ax}, a>0, b>0.
\end{equation}
if $k=0$.

Let us assume that there exists $k_1 \in \mathbb{R}^+$ and $k_2 \in \mathbb{R}$ such that
\begin{equation}\label{Prelec_15}
\lambda_1(t,T)=k_1
\end{equation}

and
\begin{equation}\label{Prelec_16}
\lambda_2(t,T)=k_2
\end{equation}

From (\ref{Prelec_14bis}) and (\ref{Prelec_15}), and remembering that by Axiom 2 $F(t,T)$ is increasing with respect to $t$, we get
\begin{equation}\label{Prelec_17}
F(t,T)=b_1(T)e^{-\frac{a_1(T)}{k_1} e^{-k_1t}}
\end{equation} 

For $k_2 \neq 0$, from (\ref{Prelec_14bis}) and (\ref{Prelec_16}), and remembering that by Axiom 1 $F(t,T)$ is decreasing with respect to $T$, we get
 
\begin{equation}\label{Prelec_18} \notag
F(t,T)=b_2(t)e^{\frac{a_2(t)}{k_2} e^{-k_2T}}
\end{equation}

By Axiom 6 we get
\begin{equation}\label{Prelec_19}\notag
F(t,0)=b_2(t)e^{\frac{a_2(t)}{k_2}}=1 
\end{equation}

and, consequently,
\begin{equation}\label{Prelec_20}\notag
b_2(t)=e^{-\frac{a_2(t)}{k_2}}
\end{equation} 

such that
\begin{equation}\label{Prelec_21}
F(t,T)=e^{-\frac{a_2(t)}{k_2}[1-e^{-k_2T}]}
\end{equation} 

Considering (\ref{Prelec_17}) 
\begin{equation}\label{Prelec_22}\notag
lim_{t\rightarrow+\infty}F(t,T)=lim_{t\rightarrow+\infty}b_1(T)e^{-\frac{a_1(T)}{k_1} e^{-k_1t}}=b_1(T)
\end{equation}

so that by Axiom 7 we get
\begin{equation}\label{Prelec_23}\notag
b_1(T)=1
\end{equation}

and consequently, considering again (\ref{Prelec_17}),
\begin{equation}\label{Prelec_24}
F(t,T)=e^{-\frac{a_1(T)}{k_1} e^{-k_1t}}
\end{equation}

From (\ref{Prelec_21}) and (\ref{Prelec_24}) we get
\begin{equation}\label{Prelec_25}\notag
e^{-\frac{a_1(T)}{k_1} e^{-k_1t}}=e^{-\frac{a_2(t)}{k_2}[1- e^{-k_2T}]}
\end{equation}

and, consequently, we obtain
\begin{equation}\label{Prelec_26}
\frac{a_1(T)}{k_1} e^{-k_1t}=\frac{a_2(t)}{k_2}[1- e^{-k_2T}]\notag
\end{equation}

that can be satisfied only if there exists $c \in \mathbb{R}^+$ such that
\begin{equation}\label{Prelec_27}
\frac{a_1(T)}{1- e^{-k_2T}}=\frac{a_2(t)}{k_2}{k_1} e^{k_1t}=c \notag
\end{equation} 

and thus 
\begin{equation}\label{Prelec_28} \notag
a_1(T)=c(1- e^{-k_2T})
\end{equation} 

so that, substituting in (\ref{Prelec_24}), 
\begin{equation}\label{Prelec_29} 
F(t,T)=e^{-c\frac{(1- e^{-k_2T})}{k_1} e^{-k_1t}}\notag
\end{equation} 

so that, taking $\delta=k_1$, $\gamma=k_2$ and $r=\frac{ck_2}{k_1}$, we re-obtain

$$\\ F(t,T) = e^{- r e^{-\delta t }\frac{1-e^{- \gamma T}}{\gamma} }$$

For $k_2=0$, from (\ref{Prelec_14ter}) and (\ref{Prelec_16}), and remembering that by Axiom 1 $F(t,T)$ is decreasing with respect to $T$, we get
 
\begin{equation}\label{Prelec_18bis} \notag
F(t,T)=b_2(t)e^{-a_2(t)T}
\end{equation}

By Axiom 6 we get
\begin{equation}\label{Prelec_19bis}\notag
F(t,0)=b_2(t)=1 
\end{equation}

and, consequently,

\begin{equation}\label{Prelec_21bis}
F(t,T)=e^{-a_2(t)T}
\end{equation} 

From (\ref{Prelec_21bis}) and (\ref{Prelec_24}) we get
\begin{equation}\label{Prelec_25bis}\notag
e^{-\frac{a_1(T)}{k_1} e^{-k_1t}}=e^{-a_2(t)T}
\end{equation}

and, consequently, we obtain
\begin{equation}\label{Prelec_26bis}
\frac{a_1(T)}{k_1} e^{-k_1t}=a_2(t)T\notag
\end{equation}

that can be satisfied only if there exists $c \in \mathbb{R}^+$ such that
\begin{equation}\label{Prelec_27bis}
\frac{a_1(T)}{T}=\frac{a_2(t)}{e^{-k_1t}}=c \notag
\end{equation} 

and thus 
\begin{equation}\label{Prelec_28bis} \notag
a_1(T)=cT
\end{equation} 

so that, substituting in (\ref{Prelec_24}), 
\begin{equation}\label{Prelec_29bis} 
F(t,T)=e^{-ce^{-k_1t}T}\notag
\end{equation} 

so that, taking $\delta=k_1$ and $r=c$, we re-obtain

$$\\ F(t,T) = e^{-re^{-\delta t} T}$$
\end{proof}

\subsection{Constant Relative Decreasing Impatience for $t$ and $T$}
In this section we consider the Pratt–Arrow relative degree of convexity of the logarithm of the discount function, namely the Prelec's relative measure known also as \emph{relative decreasing impatience}, that is defined as:

\begin{equation}\label{Prelecr_1}
\mu(\tau) = - \tau\frac{\big[ ln \ \phi(\tau)\big]''}{\big[ ln \ \phi(\tau) \big]'} \notag
\end{equation}

where $\phi: \mathbb{R}^+ \rightarrow \mathbb{R}^+$.

Considering a discounting function $F(t,T)$, we can define Prelec's relative measure to determine 
the degree of impatience both with respect to the time $t$ and with respect to the interval $T$, as
\begin{equation}\label{Prelecr_t}
\mu_1 (t,T) = - t\frac{\big[\frac{\partial^2 ln \ F(t,T)}{\partial t^2}\big]}{\big[\frac{\partial ln \ F(t,T)}{\partial t}\big]} \notag
\end{equation}

\begin{equation}\label{Prelecr_T}
\mu_2 (t,T) = - T\frac{\big[\frac{\partial^2 ln \ F(t,T)}{\partial T^2}\big]}{\big[\frac{\partial ln \ F(t,T)}{\partial T}\big]} \notag
\end{equation}


%
\subsection{Relative Decreasing Impatience for $t$ and $T$}
%


\begin{theorem} If Axioms 1,2,6 and 7 hold, then a discount function $F(t,T)$ belongs to the CRDI-CRDI family if and only if  the Prelec's relative measures with respect to $t$ and $T$ are constant, and, in particular, $\mu_1 (t,T)=h_1>1$ and $\mu_2(t,T)=h_2<1$.
\bigskip
\end{theorem}

\begin{proof} Let us start by proving that the CRDI-CRDI discount function

$$\\ F(t,T)=e^{-r t^\alpha \frac{T^{\beta+1}}{\beta+1}}$$

has constant Prelec's relative measures $\mu_1 (t,T)=h_1>1$ and $\mu_2(t,T)=h_2<1$. Indeed, since:

\begin{equation}\label{Prelec_CRDI_CRDI_2}
\frac{\partial ln \ F(t,T)}{\partial t} =   -r \alpha t^{\alpha-1} \frac{T^{\beta+1}}{\beta+1} \notag
\end{equation}

with $\alpha<0$ and $\beta>-1$

\begin{equation}\label{Prelec_CRDI_CRDI_3}
\frac{\partial^2 ln \ F(t,T)}{\partial t^2} =  -r \alpha (\alpha-1) t^{\alpha-2} \frac{T^{\beta+1}}{\beta+1}   \notag
\end{equation}

we get 

\begin{equation}\label{Prelec_CRDI_CRDI_4}
\mu_1 (t,T)=1-\alpha >1. \notag
\end{equation}

Analogously, from

\begin{equation}\label{Prelec_CRDI_CRDI_5}
\frac{\partial ln \ F(t,T)}{\partial T} = -r t^\alpha T^\beta \notag
\end{equation}

and

\begin{equation}\label{Prelec_CRDI_CRDI_6}
\frac{\partial^2 ln \ F(t,T)}{\partial T^2} = -r t^\alpha \beta T^{\beta-1} \notag
\end{equation}

we get

\begin{equation}\label{Prelec_CRDI_CRDI_7}
\mu_2(t,T) = -\beta<1. \notag
\end{equation}

Now we prove that if Axioms 1,2,6 and 7 hold and the Prelec's relative measures $\mu_1 (t,T)>1$ and $\mu_2 (t,T)<1$ are constant, then a discount function $F(t,T)$ belongs to the CRDI-CRDI family. Let us consider a function $\phi: \mathbb{R}^+\rightarrow \mathbb{R}^+$ and a constant $h$ such that
\begin{equation}\label{Prelec_CRDI_CRDI_8}
\mu (x) = - x\frac{\big[ln \ \phi(x)\big]''}{\big[ln \ \phi(x) \big]'}=h
\end{equation}
Let us consider again the function $z : \mathbb{R}^+\rightarrow \mathbb{R}$ defined as
\begin{equation}\label{Prelec_CRDI_CRDI_9}
z(x) = \big[ln \ \phi(x)\big]'
\end{equation}
so that 
\begin{equation}\label{Prelec_CRDI_CRDI_10} \notag
\mu (x) = - x\frac{z(x)'}{z(x)}
\end{equation}
and we can rewrite (\ref{Prelec_CRDI_CRDI_8}) as
\begin{equation}\label{Prelecr_11}
x\frac{z(x)'}{z(x)}=-h \notag
\end{equation}
and, consequently
\begin{equation}\label{Prelec_CRDI_CRDI_12}
\left[ ln \ |z(x)|\right]'=-h \ ln(x)' \notag
\end{equation}
the solution of which is
\begin{equation}\label{Prelec_CRDI_CRDI_13}
z(x)=ax^{-h} \mbox{ or } z(x)=-ax^{-h}, a>0.
\end{equation}
Taking into account (\ref{Prelec_CRDI_CRDI_9}) and (\ref{Prelec_CRDI_CRDI_13}) we get
\begin{equation}\label{Prelec_CRDI_CRDI_14}
\big[ln \ \phi(x)\big]'=a x^{-h} \mbox{ or } \big[ln \ \phi(x)\big]'=-a x^{-h}, a>0\notag
\end{equation}
and remembering that $\phi:\mathbb{R}^+\rightarrow \mathbb{R}^+$ is
\begin{equation}\label{Prelec_CRDI_CRDI_14bis}
\phi(x)=b e^{a\frac{x^{-h+1}}{-h+1}} \mbox{ or } \phi(x)=b e^{-a\frac{x^{-h+1}}{-h+1}}, a>0, b>0.
\end{equation}

Taking into account that $\mu_1(t,T)=h_1>1$ and $\mu_2(t,T)=h_2<1$ and Axioms 1 and 2, from (\ref{Prelec_CRDI_CRDI_14bis}) we get 
\begin{equation}\label{Prelec_CRDI_CRDI_17}
F(t,T)=b_1(T)e^{a_1(T)\frac{t^{-h_1+1}}{-h_1+1}}
\end{equation} 
and 
\begin{equation}\label{Prelec_CRDI_CRDI_18} 
F(t,T)=b_2(t)e^{-a_2(t)\frac{T^{-h_2+1}}{-h_2+1}}.
\end{equation} 
Taking into account (\ref{Prelec_CRDI_CRDI_18}), by Axiom 6 we get
\begin{equation}\label{Prelec_CRDI_CRDI_19} 
F(t,0)=b_2(t)=1 \notag
\end{equation} 
such that
\begin{equation}\label{Prelec_CRDI_CRDI_21}
F(t,T)=e^{-a_2(t)\frac{T^{-h_2+1}}{-h_2+1}}
\end{equation} 
Considering (\ref{Prelec_CRDI_CRDI_17}) 
\begin{equation}\label{Prelec_CRDI_CRDI_22}
lim_{t\rightarrow+\infty}F(t,T)=lim_{t\rightarrow+\infty}b_1(T)e^{a_1(T)\frac{t^{-h_1+1}}{-h_1+1}}=b_1(T) \notag
\end{equation} 
so that by Axiom 7 we get
\begin{equation}\label{Prelec_CRDI_CRDI_23}
b_1(T)=1 \notag
\end{equation} 
and, consequently, considering again (\ref{Prelec_CRDI_CRDI_17}),
\begin{equation}\label{Prelec_CRDI_CRDI_24}
F(t,T)=e^{a_1(T)\frac{t^{-h_1+1}}{-h_1+1}}
\end{equation} 
From (\ref{Prelec_CRDI_CRDI_21}) and (\ref{Prelec_CRDI_CRDI_24}) we get
\begin{equation}\label{Prelec_CRDI_CRDI_25}
e^{a_1(T)\frac{t^{-h_1+1}}{-h_1+1}}=e^{-a_2(t)\frac{T^{-h_2+1}}{-h_2+1}} \notag
\end{equation} 
and, thus, we obtain
\begin{equation}\label{Prelec_CRDI_CRDI_26}
a_1(T)\frac{t^{-h_1+1}}{-h_1+1}=-a_2(t)\frac{T^{-h_2+1}}{-h_2+1} \notag
\end{equation} 
that can be satisfied if and only if there exists $c>0$ such that
\begin{equation}\label{Prelec_CRDI_CRDI_27}
a_1(T)\frac{-h_2+1}{T^{-h_2+1}}=-a_2(t)\frac{-h_1+1}{t^{-h_1+1}}=c
\end{equation} 
From (\ref{Prelec_CRDI_CRDI_27}) we get 
\begin{equation}\label{Prelec_CRDI_CRDI_28}
a_1(T)=c\frac{T^{-h_2+1}}{-h_2+1} \notag
\end{equation} 
so that, substituting in (\ref{Prelec_CRDI_CRDI_24}), 
\begin{equation}\label{Prelec_CRDI_CRDI_29} \notag
F(t,T)=e^{c\frac{T^{-h_2+1}}{-h_2+1}\frac{t^{-h_1+1}}{-h_1+1}}
\end{equation} 
and, consequently, taking $\alpha=-h_1+1$, $\beta=-h_2$ and $r=-\frac{c}{-h_1+1}$, we re-obtain

$$\\ F(t,T)=e^{-r t^\alpha \frac{T^{\beta+1}}{\beta+1}}$$

\end{proof}

\subsection{Constant Absolute Decreasing Impatience for $t$ and Constant Relative Decreasing Impatience for $T$}

%


\begin{theorem} If Axioms 1,2,6 and 7 hold, then a discount function $F(t,T)$ belongs to the CADI-CRDI family if and only if  the Prelec's absolute measure with respect to $t$ and Prelec's relative  measure with respect to $T$ are constant, and, in particular, $\lambda_1 (t,T)=k_1>0$ and $\mu_2(t,T)=h_2<1$.
\bigskip
\end{theorem}

\begin{proof} Let us start by proving that the CADI-CRDI discount function

$$\\ F(t,T)=e^{-r e^{-\delta t} \frac{T^{\beta+1}}{\beta+1}}$$

has constant Prelec's absolute measure $\lambda_1 (t,T)$. Indeed, since:

\begin{equation}\label{Prelec_CADI_CRDI_2}
\frac{\partial \log F(t,T)}{\partial t} =   r \delta \frac{T^{\beta+1}}{\beta+1} e^{-\delta t} \notag
\end{equation}

and 

\begin{equation}\label{Prelec_CADI_CRDI_3}
\frac{\partial^2 \log F(t,T)}{\partial t^2} =   -r \delta^2 \frac{T^{\beta+1}}{\beta+1} e^{-\delta t}  \notag
\end{equation}

we get 

\begin{equation}\label{Prelec_CADI_CRDI_4}
\lambda_1 (t,T)=\delta>0 \notag
\end{equation}

Analogously, from

\begin{equation}\label{Prelec_CADI_CRDI_5}
\frac{\partial \log F(t,T)}{\partial T} = -r e^{-\delta t} T^\beta  \notag
\end{equation}

and

\begin{equation}\label{Prelec_CADI_CRDI_6}
\frac{\partial^2 \log F(t,T)}{\partial T^2} = -r e^{-\delta t} \beta T^{\beta-1} \notag
\end{equation}

we get

\begin{equation}\label{Prelec_CADI_CRDI_7}
\mu_2(t,T) = -\beta <1\notag
\end{equation}

Now we prove that if Axioms 1,2,6 and 7 hold and the Prelec's relative measures $\lambda_1 (t,T)=k_1$ and $\mu_2 (t,T)=h_2$ are constant with $k_1>0$ and $h_2>1$, then a discount function $F(t,T)$ belongs to the CADI-CRDI family. 
From (\ref{Prelec_24}) and (\ref{Prelec_CRDI_CRDI_21}) we get
\begin{equation}\label{Prelec_24_}
F(t,T)=e^{-\frac{a_1(T)}{k_1} e^{-k_1t}}, a_1(T)>0,
\end{equation}
and 
\begin{equation}\label{Prelec_CRDI_CRDI_21_}
F(t,T)=e^{-a_2(t)\frac{T^{-h_2+1}}{-h_2+1}}, a_2(t)>0\notag
\end{equation} 
and, consequently,
\begin{equation}\label{Prelec_CRDI_CRDI_15}
\frac{a_1(T)}{k_1} e^{-k_1t}=a_2(t)\frac{T^{-h_2+1}}{-h_2+1}
\end{equation}
which holds if and only if there exists $c>0$ such that
\begin{equation}\label{Prelec_CRDI_CRDI_27}
a_1(T)\frac{-h_2+1}{T^{-h_2+1}}=a_2(t)\frac{k_1}{e^{-k_1t}}=c
\end{equation} 
From (\ref{Prelec_CRDI_CRDI_27}) we get 
\begin{equation}\label{Prelec_CRDI_CRDI_28}
a_1(T)=c\frac{T^{-h_2+1}}{-h_2+1} \notag
\end{equation} 
so that, substituting in (\ref{Prelec_24_}), 
\begin{equation}\label{Prelec_CRDI_CRDI_29} \notag
F(t,T)=e^{-ce^{-k_1t}\frac{T^{-h_2+1}}{-h_2+1}}
\end{equation} 
and, consequently, taking $\delta=k_1$, $\beta=-h_2$ and $r=\frac{c}{-h_2+1}$, we re-obtain

$$\\ F(t,T)=e^{-r e^{-\delta t} \frac{T^{\beta+1}}{\beta+1}}$$

\end{proof}

\subsection{Constant Relative Decreasing Impatience for $t$ and Constant Absolute Decreasing Impatience $T$}

%


\begin{theorem} If Axioms 1,2,6 and 7 hold, then a discount function $F(t,T)$ belongs to the CRDI-CADI family if and only if  the Prelec's relative measure with respect to $t$ and Prelec's absolute  measure with respect to $T$ are constant, that is $\mu_1 (t,T)=h_1$ and $\lambda_2(t,T)=k_2$, with $h_1>1$ and $k_2 \in \mathbb{R}$.
\bigskip
\end{theorem}

\begin{proof} Let us start by proving that the CRDI-CADI discount function

$$\\ F(t,T)=e^{-r t^\alpha \frac{1-e^{-\gamma T}}{\gamma}}$$

has constant Prelec's relative measure $\mu_1 (t,T)$. Indeed, since:

\begin{equation}\label{Prelecar_2}
\frac{\partial \log F(t,T)}{\partial t} =  - r \alpha t^{\alpha-1}   \frac{1-e^{-\gamma T}}{\gamma} \notag
\end{equation}

and 

\begin{equation}\label{Prelecar_3}
\frac{\partial^2 \log F(t,T)}{\partial t^2} =   - r \alpha (\alpha-1) t^{\alpha-2}   \frac{1-e^{-\gamma T}}{\gamma} \notag
\end{equation}

we get 

\begin{equation}\label{Prelecar_4}
\mu_1 (t,T)= 1-\alpha>0  \notag
\end{equation}

Analogously, from

\begin{equation}\label{Prelecar_5}
\frac{\partial \log F(t,T)}{\partial T} = - r t^\alpha  e^{-\gamma T}\notag
\end{equation}

and

\begin{equation}\label{Prelecr6}
\frac{\partial^2 \log F(t,T)}{\partial T^2} = \gamma r t^\alpha  e^{-\gamma T}\notag 
\end{equation}

we get

\begin{equation}\label{Prelecar7}
\lambda_2(t,T) = \gamma \notag
\end{equation}

Now we prove that if Axioms 1,2,6 and 7 hold and the Prelec's relative measures $\mu_1 (t,T)$ and $\lambda_2 (t,T)$ are constant, then a discount function $F(t,T)$ belongs to the CRDI-CADI family. 
Let us assume that there exists $h_1, k_2 \in \mathbb{R}$ such that
\begin{equation}\label{Prelecar_15}
\lambda_1(t,T)=h_1>1
\end{equation}
and
\begin{equation}\label{Prelecar_16}
\mu_2(t,T)=k_2
\end{equation}
First, let us suppose that $k_2 \neq 0$.
Let us remember that by Axioms 1,2,6 and 7 and hypotheses (\ref{Prelecar_15}) and (\ref{Prelecar_16}) we get (\ref{Prelec_CRDI_CRDI_24}) and (\ref{Prelec_21}) that is 
\begin{equation}\label{Prelec_CRDI_CRDI_24_}
F(t,T)=e^{a_1(T)\frac{t^{-h_1+1}}{-h_1+1}}, a_1(T)>0,\notag
\end{equation} 
and 
\begin{equation}\label{Prelec_21_}
F(t,T)=e^{\frac{-a_2(t)}{k_2}[1- e^{-k_2T}]}, a_2(t)>0,\notag
\end{equation} 
and, consequently, we obtain
\begin{equation}\label{Prelecr_26}
a_1(T)\frac{t^{-h_1+1}}{-h_1+1}=-\frac{a_2(t)}{k_2}[1- e^{-k_2T}] 
\end{equation} 
that can be satisfied if and only if there exists $c>0$ such that
\begin{equation}\label{Prelecr_27}
\frac{a_1(T)k_2}{1- e^{-k_2T}}=-a_2(t)\frac{-h_1+1}{t^{-h_1+1}}=c \notag
\end{equation} 
and thus 
\begin{equation}\label{Prelecr_28} \notag
a_1(T)=c\frac{1- e^{-k_2T}}{k_2}
\end{equation} 
so that, substituting in (\ref{Prelecr_26}), we obtain
\begin{equation}\label{Prelecr_29}
F(t,T)=e^{c\frac{1- e^{-k_2T}}{k_2}\frac{t^{-h_1+1}}{-h_1+1}} \notag
\end{equation} 
and, taking $\alpha=-h_1+1$, $r=\frac{c}{h_1-1}$ and $\gamma=k_2$, we re-obtain

$$\\ F(t,T)=e^{-r t^\alpha \frac{1-e^{-\gamma T}}{\gamma}}.$$

Suppose now that $k_2 = 0$. By Axioms 1,2,6 and 7, hypotheses (\ref{Prelecar_15}) and $k_2=0$ we get (\ref{Prelec_CRDI_CRDI_24_}) and (\ref{Prelec_21}), that is 
\begin{equation}\label{Prelec_21bis_}
F(t,T)=e^{-a_2(t)T}
\end{equation} 
and, consequently, from (\ref{Prelec_CRDI_CRDI_24_}) and (\ref{Prelec_21bis_}) we obtain
\begin{equation}\label{Prelecr_26_}
a_1(T)\frac{t^{-h_1+1}}{-h_1+1}=- a_2(t) T 
\end{equation} 
that can be satisfied if and only if there exists $c>0$ such that
\begin{equation}\label{Prelecr_27_}
\frac{a_1(T)}{T}=-a_2(t)\frac{-h_1+1}{t^{-h_1+1}}=c \notag
\end{equation} 
and thus 
\begin{equation}\label{Prelecr_28_} \notag
a_1(T)=cT
\end{equation} 
so that, substituting in (\ref{Prelecr_26_}), we obtain
\begin{equation}\label{Prelecr_29_}
F(t,T)=e^{c\frac{t^{-h_1+1}}{-h_1+1}T} \notag
\end{equation} 
and, taking $\alpha=-h_1+1$ and $r=\frac{c}{h_1-1}$, we re-obtain

$$\\ F(t,T)=e^{-r t^\alpha T}.$$

\end{proof}

\newpage




\section{Experiment and discussion}\label{Experiment}

The next section presents an experiment that aims to demonstrate the feasibility of the various CADI and CRDI discount functions for intertemporal choice. 

\subsection{Method}

We conducted a time discounting experiment with undergraduate students from Department of Economics and Business at University of Catania.

\begin{itemize}
    \item[1)] First wave (May 2022). Cohort of students from the course of Financial Mathematics, 121 participants.

    \item[2)] Second wave (January 2023). Cohort of students from the course of General Mathematics, 161 participants.
\end{itemize}

Here we describe the experiment procedure considering the first wave. For the second wave, results are reported in the appendix.
One hundred twenty-one subjects took part in the experiment done in May 2022, 64 female e 57 male. The analyses were conducted on the results from 117 participants. Four results were dropped: those subjects exhibited inconsistency as they preferred the later reward to the earlier one even if the amount of the two rewards was the same.

We stimulated subjects by a non-monetary reward.
Students participating in the experiment were given the possibility to split the exam of Financial Mathematics in multiple parts. They may also take that exam anytime they want.

The experiment was run by computer. In Microsoft Forms we set the questionnaire and a link was sent to students that decide to take part in our experiment. 
Participants were informed that they would choose between pairs of amounts of money available at different times.

All payoffs in the stimuli were hypothetical. 
Studies have shown differences between real and hypothetical choice. Nevertheless most of them have concluded that the behavioral patterns are the same for cognitively quite simple tasks such as those in the experiment of this work.

Specifically, students were asked to make 43 choices between smaller amounts of money available soon and larger amounts of money available at specified times in the future. 

\begin{table}[H]
\centering
\caption{Questionnaire format}
\label{tab:Prova}
\begin{threeparttable}
\begin{tabular}{lcc}
\hline
\hline
\\[-1em]
 & Option 1 & Option 2 \\
\\[-1em]
\hline
\\[-1em]
You receive & \euro 100 & \euro 120 \\
\\[-1em]
When: & today  & 1 week from today\\
\\[-1em]
Your choice: & $\square$ & $\square$ \\
\\[-1em]
\hline
\hline
\end{tabular}
\end{threeparttable}
\end{table}


\medskip

The amount of money varies between 80 to 300 euros, the time $t$ varies between today and one year and the time delay $T$ varies between one week and one year\footnote{See the Appendix for the full-text questionnaire.}.

\subsection{Estimation of Discounting Parameters}

We use the approach adopted by Tanaka, Camerer and Nguyen \cite{tanaka2016risk} to test CADI-CADI, CRDI-CRDI, CADI-CRDI and CRDI-CADI models. As in their study, the questionnaire for the experiment is based on binary choices between monetary trade-offs. Such a format, allow us to consider a logistic function and, therefore, to run a non-linear regression in order to estimate the model parameters.

Let us denote by $x$ the immediate reward and by $y$ the delayed monetary outcome. The former is available in $t$ days, whereas the latter in $t+T$ days\footnote{For example, in Table \ref{tab:Prova} the sooner reward is available today, so $t=0$ holds. The later payment is available in $t+T=7$, as the delay is $T=7$.
}. Let $P \big( x > (y,t,T) \big)$ be the probability of choosing $x$ at $t$ time over $y$ at $t+T$, we describe this relation as\footnote{Inside the square brackets, Camerer et. al. considered also a noise parameter $\mu$, having $[- \mu \big(x-y \cdot F(t) \big) ]$.}:

\begin{eqnarray}\label{logistic_1}
P \big( x > (y,t,T) \big) = \frac{1}{1+exp \Big[- \big(x-y \cdot F(t,T) \big) \Big]}
\end{eqnarray}

Thus, we estimate for each type of discount function its peculiar parameters. In place of the discount function $F(t,T)$ we consider, from time to time, CADI-CADI, CRDI-CRDI, CADI-CRDI and finally CRDI-CADI discount function.
Table \ref{tab:Variables} shows estimation results comparing each discount function.
Hyperbolic and exponential discounting parameters were estimated as well: we adopted, respectively, the general hyperbola formulated by Loewestein and Prelec and the constant discounting proposed by Samuelson. By a nonlinear least-squares regression procedure, we fitted the logistic function (\ref{logistic_1}).

\begin{table}[H]
\centering
\caption{Comparison of Discounting models: first wave (117 subject)}
\label{tab:Variables}
\begin{threeparttable}
\begin{tabular}{lcccccc}
\hline
\hline
\\[-1em]
 & \small{CADI-CADI} & \small{CRDI-CRDI} & \small{CADI-CRDI} & \small{CRDI-CADI} & Hyperbolic & Exponential\\
\\[-1em]
\hline
\\[-1em]
$r$ & 0.0076*** & 0.0320*** & 0.0122*** & 0.0200***\\
& \smaller{(0.0001)} \ \ \ \ \ & \smaller{(0.0005)} \ \ \ \ \ & \smaller{(0.0002)} \ \ \ \ \  & \smaller{(0.0004)} \ \ \ \ \ \\
\\[-1em]
$\delta$ $\small{(\times 10^{-1})}$ & 0.0017*** & & 0.0017***\\
& \smaller{(0.00003)} \ \ \ \ \ & & \smaller{(0.00003)} \ \ \ \ \ \\
\\[-1em]
$\gamma$ & 0.0124*** & & & 0.0548*** \\
& \smaller{(0.0003)}  \ \ \ \ \ & & & \smaller{(0.0011)}  \ \ \ \ \ \\
\\[-1em]
$\alpha$ & & -0.1344***    & & 0.0635***    & 0.0167***\\
& & \smaller{(0.0037)}   \ \ \ \ \   & & \smaller{(0.0010)}    \ \ \ \ \    & \smaller{(0.0003)} \ \ \ \ \ \\
\\[-1em]
$\beta$ & & -0.4446***  & -0.2966*** & & 0.0255*** & 0.0587***\\
& & \smaller{(0.0051)}   \ \ \ \ \  &  \smaller{(0.0051)}  \ \ \ \ \  & & \smaller{(0.0004)} \ \ \ \ \  & \smaller{(0.0002)}  \ \ \ \ \  \\
& \\
\small{Observations} & 5,031 & 5,031 & 5,031 & 5,031 & 5,031 & 5,031 \\
\small{Adjusted $R^2$} & 0.6004 & 0.5383 & 0.6005 & 0.5092 & 0.5441 & 0.3058\\
\\[-1em]
\hline
\hline
\end{tabular}
 \begin{tablenotes}
            \item \emph{Notes:} Robust standard errors are in parentheses. CRDI-CRDI and CRDI-CADI discount function are adjusted with $\epsilon = 0.001$ for $t=0$.\\
            *** Significant at the 1 percent level.
 \end{tablenotes}
\end{threeparttable}
\end{table}

For the CADI-CADI model, the estimated values of $(r,\delta,\gamma)$ are $(0.0076, 0.00017,0.0124)$. This implies individuals should trade 95.03 euros today for 100 euros in a week, 66.2 euros today for 100 euros in three months, 57.86 euros today for 100 euros in six months and 54.5 euros today for 100 euros in one year. As one can see, the discount factor dramatically drops for periods close to the present. However, for moments perceived as far from today the discount factor flattens out: indeed, the latter is almost the same for six months and for one year, despite being these two periods more distant from each other than from today.

Switching to the CRDI-CRDI model, parameters $(r,\alpha,\beta)$ have values $(0.0320, -0.1344, -0.4446)$. Namely, a subject should trade 84.38 tomorrow for 100 euros in one week, 49.59 euros tomorrow for 100 euros in three months and 21.73 euros today for 100 euros in one year. In such a case, the discount structure declines more quickly for distant periods if compared to the CADI-CADI discount function.

Let us notice that the full-CRDI and the CRDI-CADI discount functions are defined for $t>0$. Since some of the binary trade-off involves the present time, i.e. $t=0$ (such in the case of the trade-off between today and one week), as regards those models and for $t=0$, we adjusted it by adding to the delay $t$ an amount $\epsilon = 0.001$, small but positive.

The full-CADI and the CADI-CRDI turn out to be the models with best performances\footnote{Considering the estimates in the second wave, from the cohort of students from General Mathematics, the models with the best performance are, again, the CADI-CRDI and the CADI-CADI discount functions.}. Either way, parameter $\delta$, that is a decay rate of the initial discount rate with respect to $t$, is significant and close to zero.
Results show parameter $r$ to be significant in each model. 
The parameter $\alpha$, that captures the degree of nonlinear scaling future time perception with respect to t, reveals to be significant in both models, but in fact it is negative in the full-CRDI model but positive in the CRDI-CADI model. Time delay $T$ parameters for CADI and CRDI, respectively $\gamma$ and $\beta$ are significant. The latter is negative in all the cases.

\begin{table}[H]
\centering
\caption{Comparison of Discounting models: second wave (158 subject)}
\label{tab:Variables_second_wave}
\begin{threeparttable}
\begin{tabular}{lcccccc}
\hline
\hline
\\[-1em]
 & \small{CADI-CADI} & \small{CRDI-CRDI} & \small{CADI-CRDI} & \small{CRDI-CADI} & Hyperbolic & Exponential \\
\\[-1em]
\hline
\\[-1em]
$r$ & 0.0152*** & 0.0323*** & 0.0122*** & 0.1097***\\
& \smaller{(0.0001)} \ \ \ \ \ & \smaller{(0.0005)} \ \ \ \ \ & \smaller{(0.0002)} \ \ \ \ \  & \smaller{(0.0051)} \ \ \ \ \ \\
\\[-1em]
$\delta$ & 0.0018*** & & 0.00018***\\
& \smaller{(0.0000)} \ \ \ \ \ & & \smaller{(0.00003)} \ \ \ \ \ \\
\\[-1em]
$\gamma$ & 0.0190*** & & & 0.3164*** \\
& \smaller{(0.0002)}  \ \ \ \ \ & & & \smaller{(0.0147)}  \ \ \ \ \ \\
\\[-1em]
$\alpha$ & & -0.1381***  & & 0.0879*** & 0.0089***\\
& & \smaller{(0.0035)}   \ \ \ \ \   & & \smaller{(0.0005)}    \ \ \ \ \    & \smaller{(0.0001)} \ \ \ \ \ \\
\\[-1em]
$\beta$ & & -0.4365***  & -0.2936*** & & 0.0262*** & 0.0590***\\
& & \smaller{(0.0046)}   \ \ \ \ \  &  \smaller{(0.0045)}  \ \ \ \ \  & & \smaller{(0.0003)} \ \ \ \ \  & \smaller{(0.0002)}  \ \ \ \ \  \\
& \\
\small{Observations} & 6,794 & 6,794 & 6,794 & 6,794 & 6,794 & 6,794 \\
\small{Adjusted $R^2$} & 0.6026 & 0.5829 & 0.6205 & 0.4216 & 0.5849 & 0.3369\\
\\[-1em]
\hline
\hline
\end{tabular}
 \begin{tablenotes}
            \item \emph{Notes:} CRDI-CRDI and CRDI-CADI adjusted with $\epsilon =0.001$ for $t=0$. Robust standard errors are in parentheses. \\
            *** Significant at the 1 percent level.
\end{tablenotes}
\end{threeparttable}
\end{table}

Considering both the first and the second wave, full-CADI and CADI-CRDI models overrule longstanding discounting models, as the Loewestein and Prelec's two-parameter hyperbolic model and the classical exponential model of Samuelson. Thus, we focus attention on full-CADI and CADI-CRDI models. 
The two above-mentioned models show high level of $R^2$, improving it in comparison with the hyperbolic and the exponential model.

\subsection{Personality traits and Time Discounting}

We then estimate a logistic function in order to analyze the relationship, if any, between how people discount time and their personality traits. 
There is a part of both economic and psychological literature that has investigated the connection between time preference and individual traits related to personality. 
Research has focused particularly on two economic traits, risk aversion and time preference, analyzing how relate with cognitive factors, such as numeracy and IQ, and psychological traits, measurable with a variety of indicators.

Results from evidence appear to be ambiguous so far. On the one hand, Dohmen et al. \cite{dohmen2010risk} for example found no significant relation between personality traits, of the Big Five type, and measures of time preference. On the other, Becker et al. \cite{becker2012relationship} finds a low association between economic preferences and psychological measures of personality. Furthermore, Borghans et al. \citep{borghans2008economics} claimed that time preference is related to conscientiousness, one component of the so-called “Big Five” personality taxonomy, mainly capturing self-discipline attitude and diligence.

In order to elicit individual personality traits, we elicited for each participant the Ten-Item Personality Inventory, \emph{TIPI}, developed by Gosling, Rentfrow, and Swann \cite{gosling2003very}. Each item consists of two descriptors that participants must rate on a 7-point Likert scale ranging from 1 (disagree strongly) to 7 (agree strongly).

The TIPI is a very short measure of the Big Five (or Five-Factor Model) dimensions. It assumes that individual differences in adult personality characteristics can be organized in terms of five broad trait domains: \emph{Extraversion (E)}, \emph{Agreeableness (A)}, \emph{Conscientiousness (C)}, \emph{Neuroticism (N)}, namely the level of emotional stability, and \emph{Openness to Experience (O)}. 
We adopted the italian revised version of the TIPI that has been proposed by Chiorri et al. \cite{chiorri2015psychometric}.

\begin{table}[H]
\centering
\caption{The Ten-Item Personality Inventory developed by Gosling, Rentfrow, Swann from data of the first wave.}
\label{tab:TIPI}
\begin{threeparttable}
\begin{tabular}{lllcc}
\hline
\hline
\\[-1em]
No. & Item & Scale & Mean & SD \\
\hline
\\[-1em]
1 & Extraverted, enthusiastic & E & 3.46 & 1.71 \\
2 & Critical, quarrelsome & A (r) & 2.52 & 1.50 \\
3 & Dependable, self-disciplined & C & 5.65 & 1.28 \\
4 & Anxious, easily upset & N & 4.75 & 1.77 \\
5 & Open to new experiences, complex & O & 5.15 & 1.33 \\
6 & Reserved, quiet & E (r) & 4.13 & 1.87 \\
7 & Sympathetic, warm & A & 5.15 & 1.57 \\
8 & Disorganized, careless & C (r) & 2.92 & 1.69 \\
9 & Calm, emotionally stable & N (r) & 4.40 & 1.70 \\
10 & Conventional, uncreative & O (r) & 4.18 & 1.42 \\
\\[-1em]
\hline
\hline
\end{tabular}
 \begin{tablenotes}
            \item \emph{Notes:} SD = Standard deviation, $(r)$ = reverse item.  C = Conscientiousness, A = Agreeableness, E = Extraversion, O = Openness, N = Neuroticism.
 \end{tablenotes}
\end{threeparttable}
\end{table}

The Big-Five factor representation was originally discovered by Tupes and Christal \cite{tupes1961recurrent} and, later, developed by Goldberg \cite{goldberg1992development}.
An elegant definition of the Big-Five domains is provided by Benet-Martinez and P. John \cite{benet1998cinco}:
\begin{quote}
Extraversion summarizes traits related to activity and energy, dominance, sociability, expressiveness, and positive emotions.
Agreeableness contrasts a prosocial orientation toward others with antagonism and includes traits such as altruism, tendermindedness, trust, and modesty. Conscientiousness describes socially prescribed impulse control that facilitates task- and goal-directed behavior. Neuroticism contrasts emotional stability with a broad range of negative affects, including anxiety, sadness, irritability, and nervous tension. Openness describes the breadth, depth, and complexity of an individual's mental and experiential life.
\end{quote}

Besides the Big Five traits, subjects were asked additional questions considering some of the behavioral questions of Rohde \cite{rohde2019measuring}.
Subjects specified their gender and whether they smoke. We get two dummy variables, equal to 1, respectively, if male and if smoker.

Next, we asked the number of hours per week devoted to physical activity and the average consumption of alcohol per week in order to elicit, respectively, the variables \emph{Sportweek} and \emph{Alcoholweek}\footnote{Alcohol week is given by the average number of glasses of wine and cans of beer consumed in a week.}.

\begin{table}[H]
\centering
\caption{Summary Statistics of Demographic and Behavioral Variables from data of the first wave}
\label{tab:var_regression}
\begin{threeparttable}
\begin{tabular}{p{4cm} c p{3cm} c p{3cm}}
\hline
\hline
\\[-1em]
Variable & OLS regressor & Mean \\
\hline
\\[-1em]
Gender & \checkmark & 47.01\% male\\
Age & & 21.4 \\
Sportweek & \checkmark & 3.48 \\
Alcoholweek & \checkmark & 1.6 \\
Smoker & \checkmark & 17.95\% \\
Monthly savings & & 29.73\% \\ 
Conscientiousness & \checkmark & 3.67 \\
Agreeableness & \checkmark & 5.32 \\
Extroversion & \checkmark & 5.38 \\
Openness & \checkmark & 3.83 \\
Neuroticism & \checkmark & 4.48 \\
\\[-1em]
\hline
\hline
\end{tabular}
\end{threeparttable}
\end{table}

In the case of the CADI-CADI model, we estimated the following logistic function:

\begin{eqnarray}\label{logistic_2}
P \big( x > (y,t,T) \big) = \frac{1}{1+exp \Big[- \big(x-y \cdot (e^{-r e^{-\delta t} (1-e^{-\gamma T}) / \gamma}) \big) \Big]}
\end{eqnarray}

where $r = r_0 + \sum r_i X_i$, $ \delta = \delta_0 + \sum \delta_i X_i$, $ \gamma = \gamma_0 + \sum \gamma_i X_i$, and the variables and their coefficients of the personality traits are given by $X_i, r_i, \delta_i$ and $\gamma_i$.

In Table \ref{tab:Reg_personality} are shown the results from regressing estimates of the CADI-CADI discount function model, allowing r, $\delta$ and $\gamma$ to depend on \emph{behavioral} variables.

Results in table \ref{tab:Reg_personality} show that, among the Big Five domains,  conscientiousness, agreeableness and openness are \emph{positively} related with patience (low $r$, initial discount rate). Extroverted people tend to exhibit more patience. They indeed have a slightly negative initial discount rate and, as times goes by, they are more prone to wait.
Apparently, neuroticism  is not related neither with impatience nor with the decay rate of time interval.
Men exhibit a little more patience at the beginning but, as time goes by, turn out to be more impatient.

\begin{table}[H]
\centering
\caption{Correlations with Time Discounting parameters (OLS)}
\label{tab:Reg_personality}
\begin{threeparttable}
\begin{tabular}{lccc}
\hline
\hline
\\[-1em]
& & CADI-CADI & \\
\cmidrule(lr){2-2}\cmidrule(lr){3-3}\cmidrule(lr){4-4}
\\[-1em]
 & $r$ & $\delta$ & $\gamma$ \\
\\[-1em]
\hline
\\[-1em]
\small{Constant $(r_0,\delta_0, \gamma_0$)} & 0.011*** & -0.164*** & 0.044***\\
& \smaller{(0.002)} & \smaller{(0.011)} & \smaller{(0.010)} \\
\\[-1em]
Gender & -0.006*** & -0.757*** & 0.273***\\
& \smaller{(0.001)} & \smaller{(0.219)} & \smaller{(0.043)}\\
\\[-1em]
Sportweek & 0.001*** & 0.086*** &  0.075*** \\
& \smaller{(0.00)} & \smaller{(0.038)} & \smaller{(0.009)} \\
\\[-1em]
Alcoholweek & -0.003*** & 1.31*** &  -0.144*** \\
& \smaller{(0.00)} & \smaller{(0.034)} &  \smaller{(0.009)}\\
\\[-1em]
Smoker & 0.050*** & -6.962***  & 2.685*** \\
& \smaller{(0.008)} & \smaller{(0.540)} &  \smaller{(0.416)}\\
\\[-1em]
Conscientiousness & 0.003*** & 0.554*** & 0.193*** \\
& \smaller{(0.00)} & \smaller{(0.116)} &  \smaller{(0.021)}\\
\\[-1em]
Agreeableness & 0.003*** & -0.381*** & 0.140*** \\
& \smaller{(0.00)} & \smaller{(0.116)} &  \smaller{(0.018)}\\
\\[-1em]
Extroversion & -0.006*** & 0.727*** & -0.341*** \\
& \smaller{(0.00)} & \smaller{(0.071)} &  \smaller{(0.031)}\\
\\[-1em]
Openness & 0.001*** & 0.037 & 0.112*** \\
& \smaller{(0.00)} & \smaller{(0.116)} &  \smaller{(0.017)}\\
\\[-1em]
Neuroticism & 0.000 & 0.086 & 0.019 \\
& \smaller{(0.00)} & \smaller{(0.079)} &  \smaller{(0.017)}\\
& \\
Observations & 5,031 & 5,031 & 5,031\\
Adjusted $R^2$ & 0.41 & 0.41 & 0.41\\
\\[-1em]
\hline
\hline
\end{tabular}
 \begin{tablenotes}
            \item \emph{Notes:} Robust standard errors are in parentheses.\\
            *** Significant at the 1 percent level.
 \end{tablenotes}
\end{threeparttable}
\end{table}

Table \ref{tab:Reg_personality} shows that being a smoker matters is related with impatience.
The economic theory of smoking suggests that individuals with a higher discount rate tend to put less weight on future disutility of addiction relative to present satisfaction from smoking, and therefore tend to smoke more \citep{kang2014time}.
According to the CADI-CADI model, those who appreciate cigarettes have a significant starting value, $r = 5\%$, of the discount rate, higher than non-smoking. Besides, this framework predict parameter delta $\delta = -6.96$, a large value but negative, revealing the smoker to be more and more heavily \emph{impatient} as time goes by. 
However, our framework does not say something about a potential causality relation between high discount rate and being a smoker.


\section{Conclusions}\label{conclusion}

We presented a family of discount function able to accommodate both time delay and time interval. From a theoretical perspective, the discounting setting here presented extends the formulation of CADI and CRDI discount functions proposed by Bleichrodt Rohde and Wakker, making discounting a function of two variables. Indeed, while the framework of Bleichrodt et al. accounts only for the delay, here we consider also time interval.

We show that for each of these two time components a discount function can exhibit constant absolute decreasing impatience (CADI) and constant relative decreasing impatience (CRDI). Therefore, we presented here four new two-variable discount function: CADI-CADI, CRDI-CRDI, CADI-CRDI and CRDI-CADI. In order to characterize such a discounting framework, we introduced a solid axiomatization setup, proving that a discount function belongs to a specific class of the CADI/CRDI family if and only if those axioms that describe such a class hold.

Furthermore, we derive for the class of the two-variable CADI/CRDI discounting the Prelec's measure, able to identify the dynamic inconsistencies of individual in intertemporal choice.

The family of CADI/CRDI discount functions is here tested by a field experiment we conducted on undergraduate students. Through a questionnaire based on a list of binary monetary trade-offs, we estimated the time preference parameters for each CADI/CRDI family of discounting models.
Then, we run a nonlinear regression by a logistic function in order to analyze the correlation between how people discount time and their personality. In order to elicit individual personality we elicited for each participant involved in the experiment the TIPI index. Our result suggest a strong connection between some specific personality traits and time preference.

\newpage
\section*{Acknowledgments}
\noindent 

\bibliographystyle{plain}
\bibliography{Full_bibliography}

\newpage
\section*{Appendix}

\begin{table}[H]
\centering
\caption{Spearman's rank correlations between Big Five scales.}
\label{tab:Spearman_1}
\begin{threeparttable}
\begin{tabular}{r | rcccc}
\hline
\hline
\\[-1em]
 & Extroversion & Conscientiousness & Agreeableness & Neuroticism & Openness \\
\hline
\\[-1em]
Extroversion & -  \ \ \ \ \ \ \ & & & \\
Conscientiousness & -0.058 \ \ \ \ \ & - & & \\
Agreeableness & -0.239*** & 0.144  & - &  \\
Neuroticism & 0.091 \ \ \ \ \ & -0.035 & 0.031 & - \\
Openness & 0.271***  & 0.050 & 0.145 & 0.085 & - \\
\hline
\hline
\end{tabular}
 \begin{tablenotes}
            \item \emph{Notes:} *** Significant at the 1 percent level.
 \end{tablenotes}
\end{threeparttable}
\end{table}



\begin{table}[H]
\centering
\caption{Correlations with Time Discounting parameters (OLS)}
\label{tab:Reg_personality_CADI-CRDI}
\begin{threeparttable}
\begin{tabular}{lccc}
\hline
\hline
\\[-1em]
& & CADI-CRDI & \\
\cmidrule(lr){2-2}\cmidrule(lr){3-3}\cmidrule(lr){4-4}
\\[-1em]
 & $r$ & $\delta$ & $\beta$ \\
\\[-1em]
\hline
\\[-1em]
\small{Constant $(r_0,\delta_0, \beta_0$)} & 0.023*** & 0.031*** & -0.050***\\
& \smaller{(0.003)} & \smaller{(0.002)} & \smaller{(0.056)} \\
\\[-1em]
Gender & -0.005*** & 0.065*** & 10.371***\\
& \smaller{(0.001)} & \smaller{(0.019)} & \smaller{(1.174)}\\
\\[-1em]
Sportweek & 0.0001*** & 0.0771*** &  -0.016 \\
& \smaller{(0.000)} & \smaller{(0.008)} & \smaller{(0.229)} \\
\\[-1em]
Alcoholweek & -0.003*** & -0.117*** &  1.463*** \\
& \smaller{(0.000)} & \smaller{(0.013)} &  \smaller{(0.311)}\\
\\[-1em]
Smoker & 0.021*** & 30.981***  & 308.188*** \\
& \smaller{(0.007)} & \smaller{(3.684)} &  \smaller{(39.334)}\\
\\[-1em]
Conscientiousness & -0.0002 & -0.293*** & -2.859*** \\
& \smaller{(0.003)} & \smaller{(0.011)} &  \smaller{(0.418)}\\
\\[-1em]
Agreeableness & 0.0001 & -0.392*** & -4.096*** \\
& \smaller{(0.000)} & \smaller{(0.026)} &  \smaller{(0.018)}\\
\\[-1em]
Extroversion & 0.001*** & 0.275*** & 0.766 \\
& \smaller{(0.000)} & \smaller{(0.030)} &  \smaller{(0.581)}\\
\\[-1em]
Openness & 0.002*** & -0.0004 & -4.168*** \\
& \smaller{(0.000)} & \smaller{(0.013)} &  \smaller{(0.578)}\\
\\[-1em]
Neuroticism & -0.001*** & 0.439*** & 5.009*** \\
& \smaller{(0.000)} & \smaller{(0.015)} &  \smaller{(0.482)}\\
& \\
Observations & 5,031 & 5,031 & 5,031\\
Adjusted $R^2$ & 0.48 & 0.48 & 0.48\\
\\[-1em]
\hline
\hline
\end{tabular}
 \begin{tablenotes}
            \item \emph{Notes:} Robust standard errors are in parentheses.\\
            *** Significant at the 1 percent level.
 \end{tablenotes}
\end{threeparttable}
\end{table}

\newpage

\begin{table}[H]
\centering
\caption{Correlations with Time Discounting parameters (OLS)}
\label{tab:Reg_personality_CRDI-CRDI}
\begin{threeparttable}
\begin{tabular}{lccc}
\hline
\hline
\\[-1em]
& & CRDI-CRDI & \\
\cmidrule(lr){2-2}\cmidrule(lr){3-3}\cmidrule(lr){4-4}
\\[-1em]
 & $r$ & $\alpha$ & $\beta$ \\
\\[-1em]
\hline
\\[-1em]
\small{Constant $(r_0,\alpha_0, \beta_0$)} & 0.018*** & 0.002 & -0.534***\\
& \smaller{(0.002)} & \smaller{(0.093)} & \smaller{(0.054)} \\
\\[-1em]
Gender & 0.089** & 0.423*** & -0.037***\\
& \smaller{(0.038)} & \smaller{(0.142)} & \smaller{(0.013)}\\
\\[-1em]
Sportweek &  -0.028*** & -0.050* &  0.011*** \\
& \smaller{(0.007)} & \smaller{(0.027)} & \smaller{(0.002)} \\
\\[-1em]
Alcoholweek & 0.026*** & 0.149*** &  -0.019*** \\
& \smaller{(0.009)} & \smaller{(0.048)} &  \smaller{(0.003)}\\
\\[-1em]
Smoker & -0.352*** & -0.124  & 0.205*** \\
& \smaller{(0.072)} & \smaller{(0.285)} &  \smaller{(0.041)}\\
\\[-1em]
Conscientiousness & -0.015 & 0.079 & 0.008 \\
& \smaller{(0.014)} & \smaller{(0.048)} &  \smaller{(0.005)}\\
\\[-1em]
Agreeableness & -0.044*** & -0.104 & 0.022*** \\
& \smaller{(0.015)} & \smaller{(0.065)} &  \smaller{(0.005)}\\
\\[-1em]
Extroversion & -0.076*** & -0.054 & 0.032*** \\
& \smaller{(0.013)} & \smaller{(0.041)} &  \smaller{(0.005)}\\
\\[-1em]
Openness & -0.018 & 0.0128 & 0.004 \\
& \smaller{(0.021)} & \smaller{(0.071)} &  \smaller{(0.007)}\\
\\[-1em]
Neuroticism & 0.006 & -0.055 & -0.005 \\
& \smaller{(0.014)} & \smaller{(0.046)} &  \smaller{(0.005)}\\
& \\
Observations & 5,031 & 5,031 & 5,031\\
Adjusted $R^2$ & 0.60 & 0.60 & 0.60\\
\\[-1em]
\hline
\hline
\end{tabular}
 \begin{tablenotes}
            \item \emph{Notes:} Robust standard errors are in parentheses. Discount function adjusted with $\epsilon=0.001$ for $t=0$. \\
            *** Significant at the 1 percent level. \\
            ** Significant at the 5 percent level. \\
            * Significant at the 10 percent level. \\
 \end{tablenotes}
\end{threeparttable}
\end{table}

\newpage

\begin{table}[H]
\centering
\caption{Correlations with Time Discounting parameters (OLS)}
\label{tab:Reg_personality_CRDI-CADI}
\begin{threeparttable}
\begin{tabular}{lccc}
\hline
\hline
\\[-1em]
& & CRDI-CADI & \\
\cmidrule(lr){2-2}\cmidrule(lr){3-3}\cmidrule(lr){4-4}
\\[-1em]
 & $r$ & $\alpha$ & $\gamma$ \\
\\[-1em]
\hline
\\[-1em]
\small{Constant $(r_0,\alpha_0, \gamma_0$)} & 0.057*** & 0.032*** & 0.302**\\
& \smaller{(0.005)} & \smaller{(0.008)} & \smaller{(0.023)} \\
\\[-1em]
Gender & -0.015 & 0.041* & -0.030\\
& \smaller{(0.015)} & \smaller{(0.024)} & \smaller{(0.070)}\\
\\[-1em]
Sportweek & 0.017*** & 0.029*** &  0.109*** \\
& \smaller{(0.004)} & \smaller{(0.005)} & \smaller{(0.017)} \\
\\[-1em]
Alcoholweek & -0.043*** & 0.034*** &  -0.172*** \\
& \smaller{(0.004)} & \smaller{(0.007)} &  \smaller{(0.019)}\\
\\[-1em]
Smoker & 0.304*** & -0.094***  & 1.143*** \\
& \smaller{(0.033)} & \smaller{(0.027)} &  \smaller{(0.130)}\\
\\[-1em]
Conscientiousness & -0.023*** & 0.036*** & -0.105*** \\
& \smaller{(0.006)} & \smaller{(0.007)} &  \smaller{(0.029)}\\
\\[-1em]
Agreeableness & 0.007 & -0.108*** & -0.032 \\
& \smaller{(0.006)} & \smaller{(0.013)} &  \smaller{(0.026)}\\
\\[-1em]
Extroversion & -0.007 & -0.135*** & -0.119*** \\
& \smaller{(0.004)} & \smaller{(0.009)} &  \smaller{(0.019)}\\
\\[-1em]
Openness & -0.002 & 0.059*** & 0.036 \\
& \smaller{(0.009)} & \smaller{(0.016)} &  \smaller{(0.039)}\\
\\[-1em]
Neuroticism & -0.035*** & -0.069*** & -0.205*** \\
& \smaller{(0.006)} & \smaller{(0.010)} &  \smaller{(0.028)}\\
& \\
Observations & 5,031 & 5,031 & 5,031\\
Adjusted $R^2$ & 0.44 & 0.44 & 0.44\\
\\[-1em]
\hline
\hline
\end{tabular}
 \begin{tablenotes}
            \item \emph{Notes:} Robust standard errors are in parentheses. Discount function adjusted with $\epsilon=0.001$ for $t=0$. \\
            *** Significant at the 1 percent level.\\
            ** Significant at the 5 percent level.\\
            * Significant at the 10 percent level.
 \end{tablenotes}
\end{threeparttable}
\end{table}

\newpage


\bigskip

\begin{itemize}
    \item Loewenstein and Prelec's model as function of two variables:
\begin{eqnarray}\label{LoewPrelec} \notag
F(t,T) = \Bigg[ \frac{1+\alpha t}{1+ \alpha(t+T)} \Bigg]^{\frac{\beta}{\alpha}}
\end{eqnarray}
\end{itemize}

\begin{table}[H]
\centering
\caption{Correlations with Time Discounting parameters (OLS)}
\label{tab:Reg_personality_Loewenstein_Prelec}
\begin{threeparttable}
\begin{tabular}{lcc}
\hline
\hline
\\[-1em]
& Loewenstein and Prelec model & \\
\cmidrule(lr){2-2}\cmidrule(lr){3-3}
\\[-1em]
 & $\alpha$ & $\beta$  \\
\\[-1em]
\hline
\\[-1em]
\small{Constant $(\alpha_0, \beta_0$)} & -0.035*** & 0.014***\\
& \smaller{(0.006)} & \smaller{(0.001)} \\
\\[-1em]
Gender & -0.020*** & -0.005***\\
& \smaller{(0.003)} & \smaller{(0.001)}\\
\\[-1em]
Sportweek (x$10^{-3}$) & 0.021 & -0.378*** \\
& \smaller{(0.000)} & \smaller{(0.000)} \\
\\[-1em]
Alcoholweek & -0.002*** & -0.001*** \\
& \smaller{(0.000)} & \smaller{(0.000)}\\
\\[-1em]
Smoker & -0.019*** & 0.007***\\
& \smaller{(0.003)} & \smaller{(0.001)} \\
\\[-1em]
Conscientiousness & 0.011*** & 0.003*** \\
& \smaller{(0.001)} & \smaller{(0.000)} \\
\\[-1em]
Agreeableness & 0.006*** & 0.001***  \\
& \smaller{(0.001)} & \smaller{(0.000)}\\
\\[-1em]
Extroversion & 0.005*** & -0.001***\\
& \smaller{(0.001)} & \smaller{(0.000)}\\
\\[-1em]
Openness & -0.013*** & -0.003*** \\
& \smaller{(0.001)} & \smaller{(0.000)}\\
\\[-1em]
Neuroticism & 0.022*** & 0.001*** \\
& \smaller{(0.001)} & \smaller{(0.000)}\\
& \\
Observations & 5,031 & 5,031\\
Adjusted $R^2$ & 0.43 & 0.43\\
\\[-1em]
\hline
\hline
\end{tabular}
 \begin{tablenotes}
            \item \emph{Notes:} Robust standard errors are in parentheses.\\
            *** Significant at the 1 percent level.\\
            ** Significant at the 5 percent level.\\
            * Significant at the 10 percent level.
 \end{tablenotes}
\end{threeparttable}
\end{table}


\newpage

\bigskip

\begin{itemize}
    \item Samuelson's model:
\begin{eqnarray}\label{Samuelson} \notag
F(t) = e^{- \beta  t}
\end{eqnarray}
\end{itemize}

\begin{table}[H]
\centering
\caption{Correlations with Time Discounting parameters (OLS)}
\label{tab:Reg_personality_Samuelson}
\begin{threeparttable}
\begin{tabular}{lcc}
\hline
\hline
\\[-1em]
& Samuelson model (one parameter) \\
\cmidrule(lr){2-2}
\\[-1em]
 & $\beta$  \\
\\[-1em]
\hline
\\[-1em]
Constant $(\beta_0$) & 0.030*** \\
& \smaller{(0.001)} \\
\\[-1em]
Gender (x$10^{-4}$) & -1.198\\
& \smaller{(0.000)} \\
\\[-1em]
Sportweek (x$10^{-3}$) & -0.185 \\
& \smaller{(0.000)}  \\
\\[-1em]
Alcoholweek & -0.001***  \\
& \smaller{(0.000)} \\
\\[-1em]
Smoker & 0.033*** \\
& \smaller{(0.001)}  \\
\\[-1em]
Conscientiousness (x$10^{-4}$) & 1.714  \\
& \smaller{(0.000)}  \\
\\[-1em]
Agreeableness (x$10^{-4}$) & 1.995*  \\
& \smaller{(0.001)} \\
\\[-1em]
Extroversion (x$10^{-4}$) & -3.272*** \\
& \smaller{(0.000)} \\
\\[-1em]
Openness (x$10^{-4}$) & -2.227*   \\
& \smaller{(0.000)} \\
\\[-1em]
Neuroticism (x$10^{-4}$) & -0.133  \\
& \smaller{(0.000)} \\
& \\
Observations & 5,031  \\
Adjusted $R^2$ & 0.30 \\
\\[-1em]
\hline
\hline
\end{tabular}
 \begin{tablenotes}
            \item \emph{Notes:} Robust standard errors are in parentheses.\\
            *** Significant at the 1 percent level.\\
            ** Significant at the 5 percent level.\\
            * Significant at the 10 percent level.
 \end{tablenotes}
\end{threeparttable}
\end{table}

\newpage
\bigskip

\begin{itemize}
    \item Parameters of students attending the General Mathematics' course.
\end{itemize}



\begin{table}[H]
\centering
\caption{Correlations with Time Discounting parameters (OLS)}
\label{tab:Reg_personality_General_CADI-CADI}
\begin{threeparttable}
\begin{tabular}{lccc}
\hline
\hline
\\[-1em]
& & CADI-CADI & \\
\cmidrule(lr){2-2}\cmidrule(lr){3-3}\cmidrule(lr){4-4}
\\[-1em]
 & $r$ & $\delta$ (x$10^{-3}$) & $\gamma$ \\
\\[-1em]
\hline
\\[-1em]
\small{Constant $(r_0,\delta_0, \gamma_0$)} & 0.018*** & -0.036 & 0.045***\\
& \smaller{(0.002)} & \smaller{(0.000)} & \smaller{(0.005)} \\
\\[-1em]
Gender & 0.001*** & 0.0691*** & 0.022**\\
& \smaller{(0.001)} & \smaller{(0.000)} & \smaller{(0.011)}\\
\\[-1em]
Sportweek & -0.001*** & 20.822** &  -0.021*** \\
& \smaller{(0.000)} & \smaller{(0.012)} & \smaller{(0.002)} \\
\\[-1em]
Alcoholweek & 0.002*** & -0.033*** &  0.044*** \\
& \smaller{(0.000)} & \smaller{(0.016)} &  \smaller{(0.004)}\\
\\[-1em]
Smoker & 0.001 & 0.0495**  & 0.008 \\
& \smaller{(0.001)} & \smaller{(0.000)} &  \smaller{(0.013)}\\
\\[-1em]
Conscientiousness & 0.001*** & 0.011 & 0.026*** \\
& \smaller{(0.0001)} & \smaller{(0.000)} &  \smaller{(0.004)}\\
\\[-1em]
Agreeableness & -0.002*** & 3.693 & -0.039*** \\
& \smaller{(0.000)} & \smaller{(0.023)} &  \smaller{(0.005)}\\
\\[-1em]
Extroversion & 0.002*** & -7.783 & 0.037*** \\
& \smaller{(0.000)} & \smaller{(0.021)} &  \smaller{(0.005)}\\
\\[-1em]
Openness & -0.0003 & -0.013 & -0.006 \\
& \smaller{(0.000)} & \smaller{(0.024)} &  \smaller{(0.0.005)}\\
\\[-1em]
Neuroticism & -0.001*** & 0.027*** & -0.024*** \\
& \smaller{(0.000)} & \smaller{(0.020)} &  \smaller{(0.004)}\\
& \\
Observations & 6,794 & 6,794 & 6,794\\
Adjusted $R^2$ & 0.60 & 0.60 & 0.60\\
\\[-1em]
\hline
\hline
\end{tabular}
 \begin{tablenotes}
            \item \emph{Notes:} Robust standard errors are in parentheses.\\
            *** Significant at the 1 percent level.\\
            ** Significant at the 5 percent level.\\
            * Significant at the 10 percent level.
 \end{tablenotes}
\end{threeparttable}
\end{table}

\newpage

\begin{table}[H]
\centering
\caption{Correlations with Time Discounting parameters (OLS)}
\label{tab:Reg_personality_General_CADI-CRDI}
\begin{threeparttable}
\begin{tabular}{lccc}
\hline
\hline
\\[-1em]
& & CADI-CRDI & \\
\cmidrule(lr){2-2}\cmidrule(lr){3-3}\cmidrule(lr){4-4}
\\[-1em]
 & $r$ & $\delta$ & $\beta$ \\
\\[-1em]
\hline
\\[-1em]
\small{Constant $(r_0,\delta_0, \beta_0$)} & 0.026*** & -0.017*** & -0.806***\\
& \smaller{(0.001)} & \smaller{(0.004)} & \smaller{(0.009)} \\
\\[-1em]
Gender & -0.354*** & -1.041*** & -0.029***\\
& \smaller{(0.034)} & \smaller{(0.155)} & \smaller{(0.003)} \\
\\[-1em]
Sportweek & 0.013** & 0.023 &  0.001*** \\
& \smaller{(0.005)} & \smaller{(0.019)} & \smaller{(0.000)} \\
\\[-1em]
Alcoholweek & -0.067*** & -0.106** &  0.005*** \\
& \smaller{(0.006)} & \smaller{(0.047)} &  \smaller{(0.001)}\\
\\[-1em]
Smoker & 0.808*** & 1.219***  & 0.074*** \\
& \smaller{(0.032)} & \smaller{(0.169)} &  \smaller{(0.005)}\\
\\[-1em]
Conscientiousness & -0.112*** & -0.196*** & -0.010*** \\
& \smaller{(0.010)} & \smaller{(0.048)} &  \smaller{(0.001)}\\
\\[-1em]
Agreeableness & -0.180*** & -0.261*** & -0.013*** \\
& \smaller{(0.014)} & \smaller{(0.056)} &  \smaller{(0.001)}\\
\\[-1em]
Extroversion & 0.095*** & 0.048 & 0.007*** \\
& \smaller{(0.015)} & \smaller{(0.054)} &  \smaller{(0.001)}\\
\\[-1em]
Openness & 0.054*** & 0.059 & 0.005*** \\
& \smaller{(0.013)} & \smaller{(0.058)} &  \smaller{(0.001)}\\
\\[-1em]
Neuroticism & -0.134*** & -0.218*** & -0.010*** \\
& \smaller{(0.011)} & \smaller{(0.040)} &  \smaller{(0.001)}\\
& \\
Observations & 6,794 & 6,794 & 6,794 \\
Adjusted $R^2$ & 0.40 & 0.40 & 0.40  \\
\\[-1em]
\hline
\hline
\end{tabular}
 \begin{tablenotes}
            \item \emph{Notes:} Robust standard errors are in parentheses.\\
            *** Significant at the 1 percent level.
 \end{tablenotes}
\end{threeparttable}
\end{table}

\newpage

\begin{table}[H]
\centering
\caption{Correlations with Time Discounting parameters (OLS)}
\label{tab:Reg_personality_General_CRDI-CRDI}
\begin{threeparttable}
\begin{tabular}{lccc}
\hline
\hline
\\[-1em]
& & CRDI-CRDI & \\
\cmidrule(lr){2-2}\cmidrule(lr){3-3}\cmidrule(lr){4-4}
\\[-1em]
 & $r$ ($\times 10^{-2}$) & $\alpha$ ($\times 10^{-2}$) & $\beta$ \\
\\[-1em]
\hline
\\[-1em]
\small{Constant $(r_0,\alpha_0, \beta_0$)} & 1.366*** & 0.964** & -0.332***\\
& \smaller{(0.001)} & \smaller{(0.004)} & \smaller{(0.035)} \\
\\[-1em]
Gender & -0.015 & -0.321 & 0.008\\
& \smaller{(0.000)} & \smaller{(0.001)} & \smaller{(0.009)}\\
\\[-1em]
Sportweek & 0.003 & 0.019 &  -0.002 \\
& \smaller{(0.000)} & \smaller{(0.000)} & \smaller{(0.001)} \\
\\[-1em]
Alcoholweek & 0.010 & -0.025 &  -0.003 \\
& \smaller{(0.000)} & \smaller{(0.000)} &  \smaller{(0.002)}\\
\\[-1em]
Smoker & 0.020 & -0.164  & -0.001 \\
& \smaller{(0.000)} & \smaller{(0.001)} &  \smaller{(0.009)}\\
\\[-1em]
Conscientiousness & -0.006 & -0.083** & 0.002 \\
& \smaller{(0.000)} & \smaller{(0.000)} &  \smaller{(0.004)}\\
\\[-1em]
Agreeableness & -0.018 & -0.036 & 0.003 \\
& \smaller{(0.000)} & \smaller{(0.000)} &  \smaller{(0.004)}\\
\\[-1em]
Extroversion & 0.000 & -0.095** & 0.003 \\
& \smaller{(0.000)} & \smaller{(0.000)} &  \smaller{(0.004)}\\
\\[-1em]
Openness & -0.007 & 0.024 & 0.001 \\
& \smaller{(0.000)} & \smaller{(0.001)} &  \smaller{(0.004)}\\
\\[-1em]
Neuroticism & -0.019 & -0.003 & 0.007** \\
& \smaller{(0.000)} & \smaller{(0.000)} &  \smaller{(0.003)}\\
& \\
Observations & 6,794 & 6,794 & 6,794\\
Adjusted $R^2$ & 0.62 & 0.62 & 0.62\\
\\[-1em]
\hline
\hline
\end{tabular}
 \begin{tablenotes}
            \item \emph{Notes:} Robust standard errors are in parentheses. Discount function adjusted with $\epsilon = 0.001$ for $t = 0$.\\
            *** Significant at the 1 percent level. \\
            ** Significant at the 5 percent level. \\
            * Significant at the 10 percent level.
 \end{tablenotes}
\end{threeparttable}
\end{table}

\newpage


\begin{table}[H]
\centering
\caption{Correlations with Time Discounting parameters (OLS)}
\label{tab:Reg_personality_General_CRDI-CADI}
\begin{threeparttable}
\begin{tabular}{lccc}
\hline
\hline
\\[-1em]
& & CRDI-CADI & \\
\cmidrule(lr){2-2}\cmidrule(lr){3-3}\cmidrule(lr){4-4}
\\[-1em]
 & $r$ & $\alpha$ ($\times 10^{-2}$) & $\gamma$ \\
\\[-1em]
\hline
\\[-1em]
\small{Constant $(r_0,\alpha_0, \gamma_0$)} & 0.101*** & -0.317 & 0.529***\\
& \smaller{(0.015)} & \smaller{(0.004)} & \smaller{(0.079)} \\
\\[-1em]
Gender & 0.011*** & -0.011 & 0.523***\\
& \smaller{(0.003)} & \smaller{(0.001)} & \smaller{(0.019)}\\
\\[-1em]
Sportweek & -0.003*** & 0.021 &  -0.017*** \\
& \smaller{(0.001)} & \smaller{(0.000)} & \smaller{(0.003)} \\
\\[-1em]
Alcoholweek & 0.004*** & 0.014 &  0.020*** \\
& \smaller{(0.001)} & \smaller{(0.000)} &  \smaller{(0.007)}\\
\\[-1em]
Smoker & -0.011*** & -0.337***  & -0.065*** \\
& \smaller{(0.003)} & \smaller{(0.001)} &  \smaller{(0.018)}\\
\\[-1em]
Conscientiousness & -0.010*** & -0.085** & -0.054*** \\
& \smaller{(0.002)} & \smaller{(0.000)} &  \smaller{(0.009)}\\
\\[-1em]
Agreeableness & 0.000 & 0.005 & 0.0003 \\
& \smaller{(0.001)} & \smaller{(0.000)} &  \smaller{(0.007)}\\
\\[-1em]
Extroversion & -0.002 & -0.036 & -0.014* \\
& \smaller{(0.002)} & \smaller{(0.000)} &  \smaller{(0.008)}\\
\\[-1em]
Openness & 0.003** & 0.182*** & 0.021** \\
& \smaller{(0.002)} & \smaller{(0.001)} &  \smaller{(0.008)}\\
\\[-1em]
Neuroticism & 0.002** & -0.059* & 0.011* \\
& \smaller{(0.001)} & \smaller{(0.003)} &  \smaller{(0.006)}\\
& \\
Observations & 6,794 & 6,794 & 6,794 \\
Adjusted $R^2$ & 0.48 & 0.48 & 0.48\\
\\[-1em]
\hline
\hline
\end{tabular}
 \begin{tablenotes}
            \item \emph{Notes:} Robust standard errors are in parentheses. Discount function adjusted with $\epsilon = 0.001$ for $t = 0$.\\
            *** Significant at the 1 percent level.\\
            ** Significant at the 5 percent level.\\
            * Significant at the 10 percent level.
 \end{tablenotes}
\end{threeparttable}
\end{table}

\newpage

\begin{table}[H]
\centering
\caption{Correlations with Time Discounting parameters (OLS)}
\label{tab:Reg_personality_General_Loewenstein_Prelec}
\begin{threeparttable}
\begin{tabular}{lcc}
\hline
\hline
\\[-1em]
& Loewenstein and Prelec model & \\
\cmidrule(lr){2-2}\cmidrule(lr){3-3}
\\[-1em]
 & $\alpha$ & $\beta$  \\
\\[-1em]
\hline
\\[-1em]
\small{Constant $(\alpha_0, \beta_0$)} & 0.021*** & 0.036***\\
& \smaller{(0.002)} & \smaller{(0.002)} \\
\\[-1em]
Gender & 0.003*** & 0.004***\\
& \smaller{(0.001)} & \smaller{(0.001)}\\
\\[-1em]
Sportweek (x$10^{-3}$) & -0.713*** & -0.824*** \\
& \smaller{(0.000)} & \smaller{(0.000)} \\
\\[-1em]
Alcoholweek & 0.001*** & 0.002*** \\
& \smaller{(0.000)} & \smaller{(0.000)}\\
\\[-1em]
Smoker & 0.005*** & 0.006***\\
& \smaller{(0.006)} & \smaller{(0.001)} \\
\\[-1em]
Conscientiousness & -0.003*** & 0.004*** \\
& \smaller{(0.000)} & \smaller{(0.000)} \\
\\[-1em]
Agreeableness & 0.001*** & 0.001***  \\
& \smaller{(0.000)} & \smaller{(0.000)}\\
\\[-1em]
Extroversion (x$10^{-3}$) & 0.039 & 0.083\\
& \smaller{(0.000)} & \smaller{(0.000)}\\
\\[-1em]
Openness & -0.001*** & -0.002*** \\
& \smaller{(0.000)} & \smaller{(0.000)}\\
\\[-1em]
Neuroticism & 0.001*** & 0.002*** \\
& \smaller{(0.000)} & \smaller{(0.000)}\\
& \\
Observations & 6,794 & 6,794\\
Adjusted $R^2$ & 0.60 & 0.60\\
\\[-1em]
\hline
\hline
\end{tabular}
 \begin{tablenotes}
            \item \emph{Notes:} Robust standard errors are in parentheses.\\
            *** Significant at the 1 percent level.\\
            ** Significant at the 5 percent level.\\
            * Significant at the 10 percent level.
 \end{tablenotes}
\end{threeparttable}
\end{table}

\newpage

\begin{table}[H]
\centering
\caption{Correlations with Time Discounting parameters (OLS)}
\label{tab:Reg_personality_Genearal_Samuelson}
\begin{threeparttable}
\begin{tabular}{lcc}
\hline
\hline
\\[-1em]
& Samuelson model (one parameter) \\
\cmidrule(lr){2-2}
\\[-1em]
 & $\beta$ (x$10^{-3}$) \\
\\[-1em]
\hline
\\[-1em]
Constant $(\beta_0$) & 58.386*** \\
& \smaller{(0.001)} \\
\\[-1em]
Gender & 0.286 \\
& \smaller{(0.000)} \\
\\[-1em]
Sportweek & -0.054 \\
& \smaller{(0.000)}  \\
\\[-1em]
Alcoholweek & 0.122  \\
& \smaller{(0.000)} \\
\\[-1em]
Smoker & 0.263 \\
& \smaller{(0.000)}  \\
\\[-1em]
Conscientiousness & 0.066  \\
& \smaller{(0.000)}  \\
\\[-1em]
Agreeableness & -0.059  \\
& \smaller{(0.000)} \\
\\[-1em]
Extroversion & 0.174 \\
& \smaller{(0.000)} \\
\\[-1em]
Openness  & -0.181   \\
& \smaller{(0.000)} \\
\\[-1em]
Neuroticism & 0.123  \\
& \smaller{(0.000)} \\
& \\
Observations & 6,794  \\
Adjusted $R^2$ & 0.34 \\
\\[-1em]
\hline
\hline
\end{tabular}
 \begin{tablenotes}
            \item \emph{Notes:} Robust standard errors are in parentheses.\\
            *** Significant at the 1 percent level.\\
            ** Significant at the 5 percent level.\\
            * Significant at the 10 percent level.
 \end{tablenotes}
\end{threeparttable}
\end{table}

\newpage

\end{document}